\documentclass[manuscript,nonacm]{acmart}

\usepackage{booktabs} 

\usepackage{tabularx}
\usepackage{multirow}
\usepackage{mathtools}
\usepackage{nicefrac}
\usepackage{amsfonts,amsthm}
\usepackage{dsfont}
\usepackage{url}
\usepackage{lipsum}
\usepackage[htt]{hyphenat}
\usepackage{commath}
\usepackage{algorithm}
\usepackage[inline]{enumitem}
\usepackage[justification=centering]{caption}
\usepackage[noend]{algpseudocode}
\usepackage{setspace}

\usepackage{subcaption}
\usepackage[super]{nth}
\usepackage{color-edits}
\addauthor[Kartik]{kl}{blue}
\algnewcommand\algorithmicparfor{\textbf{for}}
\algnewcommand\algorithmicpardo{\textbf{do in parallel}}
\algnewcommand\algorithmicforeach{\textbf{for each}}
\algnewcommand{\IfThenElse}[3]{
  \State \algorithmicif\ #1\ \algorithmicthen\ #2\ \algorithmicelse\ #3}

\algrenewtext{ParFor}[1]{\algorithmicparfor\ #1\ \algorithmicpardo}
\algrenewtext{EndParFor}{}

\algrenewtext{ParForEach}[1]{\algorithmicforeach\ #1\ \algorithmicpardo}
\algrenewtext{EndParForEach}{}

\algrenewtext{ForEach}[1]{\algorithmicforeach\ #1\ \algorithmicfordo}
\algrenewtext{EndForEach}{}
\algdef{SnE}[PARFOR]{ParFor}{EndParFor}[1]{\algorithmicparfor\ #1\ \algorithmicpardo}%
\algdef{SnE}[PARFOREACH]{ParForEach}{EndParForEach}[1]{\algorithmicforeach\ #1\ \algorithmicpardo}%

\algdef{SnE}[FOREACH]{ForEach}{EndForEach}[1]{\algorithmicforeach\ #1}%
\algdef{SE}[DOWHILE]{Do}{doWhile}{\algorithmicdo}[1]{\algorithmicwhile\ #1}%
\algdef{SnE}[WHILEWODO]{Whilewodo}{EndWhilewodo}[1]{\algorithmicwhile\ #1 \{\}}%

\makeatletter
\renewcommand{\Function}[2]{%
  \csname ALG@cmd@\ALG@L @Function\endcsname{#1}{#2}%
  \def\jayden@currentfunction{#1}%
}
\newcommand{\funclabel}[1]{%
  \@bsphack
  \protected@write\@auxout{}{%
    \string\newlabel{#1}{{\jayden@currentfunction}{\thepage}}%
  }%
  \@esphack
}
\makeatother

\newtheorem{theorem}{Theorem}
\newtheorem{lemma}{Lemma}

\newtheorem{definition}{Definition}
\newtheorem{property}{Property}

\newcommand{\beq}{\begin{equation}}
\newcommand{\eeq}{\end{equation}}
\newcommand{\bea}{\begin{eqnarray}}
\newcommand{\eea}{\end{eqnarray}}

\setcopyright{none}
\begin{document}

\title[Parallel Peeling of Bipartite Networks]{Parallel 
Peeling of Bipartite Networks for Hierarchical Dense Subgraph Discovery
}


\author{Kartik Lakhotia}
\affiliation{%
  \institution{University of Southern California}
   \country{USA} 
 }
\email{klakhoti@usc.edu}

\author{Rajgopal Kannan}
\affiliation{%
  \institution{US Army Research Lab}
    \country{USA}
}
\email{rajgopal.kannan.civ@mail.Mil}

\author{Viktor Prasanna}
\affiliation{%
  \institution{University of Southern California}
    \country{USA}
}
\email{prasanna@usc.edu}

\begin{abstract}
    Motif-based graph decomposition is widely used to mine hierarchical 
    dense structures in graphs.
    In bipartite graphs, \textit{wing} and \textit{tip} decomposition 
    construct a hierarchy of butterfly~(2,2-biclique) dense edge 
    and vertex induced subgraphs, respectively. They have applications
    in several domains including e-commerce, recommendation systems and document analysis.
    
    Existing decomposition algorithms use a bottom-up approach that
    constructs the hierarchy in an increasing order of subgraph density.
    They iteratively select the entities~(edges or vertices) with minimum 
    support~(butterfly count) and \emph{peel} them i.e. remove them from them graph and update the support of other entities. 
    The amount of butterflies in real-world bipartite graphs makes
    bottom-up peeling computationally demanding. 
    Furthermore, the strict order of peeling entities results in a large
    number of iterations with sequential dependencies on preceding 
    support updates. 
    Consequently, parallel algorithms based on bottom up peeling
    can only utilize
    intra-iteration parallelism and require heavy synchronization, leading to poor scalability.
    

    

    In this paper, we propose a novel Parallel Bipartite Network peelinG (PBNG) framework which adopts a two-phased peeling approach to relax the order of peeling, and in turn, 
    dramatically reduce synchronization. 
    The first phase divides the decomposition hierarchy into few partitions, 
    and requires little synchronization to compute such partitioning.
    The second phase concurrently processes all of these partitions to generate 
    individual levels in the final decomposition hierarchy, and 
    requires no global synchronization. 
    Effectively, both phases of PBNG parallelize computation
    across multiple levels of decomposition hierarchy,
    which is not possible with bottom-up peeling.
    The two-phased peeling further enables batching optimizations that 
    dramatically improve the computational efficiency of PBNG. 
    The proposed approach represents a non-trivial
    generalization of our prior work on a two-phased vertex peeling 
    algorithm~\cite{lakhotia2020receipt}, and its adoption for 
    both tip and wing decomposition. 
    
    We empirically evaluate PBNG using several real-world bipartite graphs 
    and demonstrate radical improvements over the existing approaches. On a 
    shared-memory $36$ core server,
    PBNG achieves up to $19.7\times$ self-relative parallel 
    speedup. Compared to the state-of-the-art parallel framework
    P\textsc{ar}B\textsc{utterfly}, PBNG reduces synchronization 
    by up to $15260\times$ and execution time by up to $295\times$. Furthermore, it achieves up to $38.5\times$ 
    speedup over state-of-the-art algorithms specifically tuned
    for wing decomposition.
    We also present the first decomposition results of some 
    of the largest public real-world datasets, which 
    PBNG can peel in few minutes/hours, but 
    algorithms in current practice fail to process even in several days.
    Our source code is made available at  \url{https://github.com/kartiklakhotia/RECEIPT}.

\end{abstract}



%


%
%


\maketitle
\pagestyle{empty}

\renewcommand{\shortauthors}{Lakhotia et al.}

\section{Introduction}
A bipartite graph $G(W=(U, V), E)$ is a special graph whose vertices can be partitioned into
two disjoint sets $U(G)$ and $V(G)$ such that any edge $e\in E(G)$ connects a vertex from set $U(G)$ 
with a vertex from set $V(G)$. 
Several real-world systems naturally exhibit bipartite relationships, such as consumer-product purchase network of an e-commerce website~\cite{consumerProduct}, user-ratings data in a recommendation system~\cite{he2016ups, lim2010detecting}, author-paper network of a scientific field~\cite{authorPaper}, group memberships in a social network~\cite{orkut} etc.
Due to the rapid growth of data produced in these domains, efficient 
mining of dense structures in bipartite graphs has become a popular research topic~\cite{wangButterfly, wangBitruss, zouBitruss, sariyucePeeling, shiParbutterfly,lakhotia2020receipt}. 

Nucleus decomposition is commonly used to mine hierarchical dense subgraphs where 
minimum clique participation of an edge in a subgraph determines its level in the hierarchy~\cite{sariyuce2015finding}. 
Truss decomposition is arguably the most popular 
case of nucleus decomposition which uses triangles ($3$-cliques) to measure 
subgraph density~\cite{spamDet, graphChallenge, trussVLDB, sariyuce2016fast, bonchi2019distance, wen2018efficient}. However, truss decomposition is not directly applicable for bipartite graphs as they do not have triangles. One way to circumvent this issue is to
compute unipartite projection of a bipartite graph $G$
which contains an edge between each pair of vertices with common neighbor(s) in $G$. 
But this approach suffers from (a) information loss which can impact quality
of results, and (b) explosion in dataset size which can 
restrict its scalability~\cite{sariyucePeeling}. 

Butterfly ($2,2-$biclique/quadrangle) is the smallest cohesive motif in bipartite 
graphs. Butterflies can be used to directly analyze bipartite graphs and have drawn significant research 
interest in the recent years~\cite{wangRectangle,sanei2018butterfly,sanei2019fleet,shiParbutterfly, wangButterfly, he2021exploring, sariyucePeeling, wang2018efficient}. Sariyuce and Pinar~\cite{sariyucePeeling} use butterflies as a density indicator to define 
the notion of \textit{$k-$wings}
and \textit{$k-$tips}, as maximal bipartite subgraphs where each edge and vertex, 
respectively, is involved in at least $k$ butterflies.
For example, the graph shown in fig.\ref{fig:wingDemo}a
is a $1-$wing since each edge participates in at least
one butterfly.
Analogous to $k-$ trusses~\cite{cohen2008trusses}, $k-$wings~($k-$tips) represent hierarchical dense structures in the sense that a $(k+1)-$wing~($(k+1)-$tip) is a subgraph of a $k-$wing~($k-$tip).

In this paper, we explore parallel algorithms for \textit{wing}\footnote{$k-$wing and wing decomposition 
are also known as $k-$bitruss and bitruss decomposition, respectively.}  and \textit{tip decomposition} analytics,
that construct the entire hierarchy of $k-$wings
and $k-$tips in a bipartite graph,
respectively. 
For space-efficient representation of
the hierarchy, these analytics output \emph{wing number}
of each edge $e$ or \emph{tip number} of each vertex 
$u$, which represent the densest level of
hierarchy that contains $e$ or $u$, respectively. 
Wing and tip
decomposition have several real-world applications such as:
\begin{itemize}[leftmargin=*]
    \item Link prediction in recommendation systems 
    or e-commerce websites that contain
    communities of users with
    common preferences or  purchase history~\cite{he2021exploring,leicht2006vertex,navlakha2008graph,communityDet}.
    
    \item Mining nested 
    communities in  social 
    networks or discussion forums, where users affiliate with broad groups
    and more specific sub-groups 
    based on their interests.
    ~\cite{he2021exploring}.
    
    \item Detecting spam reviewers that collectively rate selected
    items in rating
    networks~\cite{mukherjee2012spotting,fei2013exploiting,lim2010detecting}.
    
    \item Document clustering
    by mining co-occurring
    keywords and groups of documents containing them~\cite{dhillon2001co}.
    
    \item Finding nested groups of researchers from
    author-paper networks~\cite{sariyucePeeling} with varying degree of collaboration.
\end{itemize}

\begin{figure}[htbp]
    \centering
\includegraphics[width=0.9\linewidth]{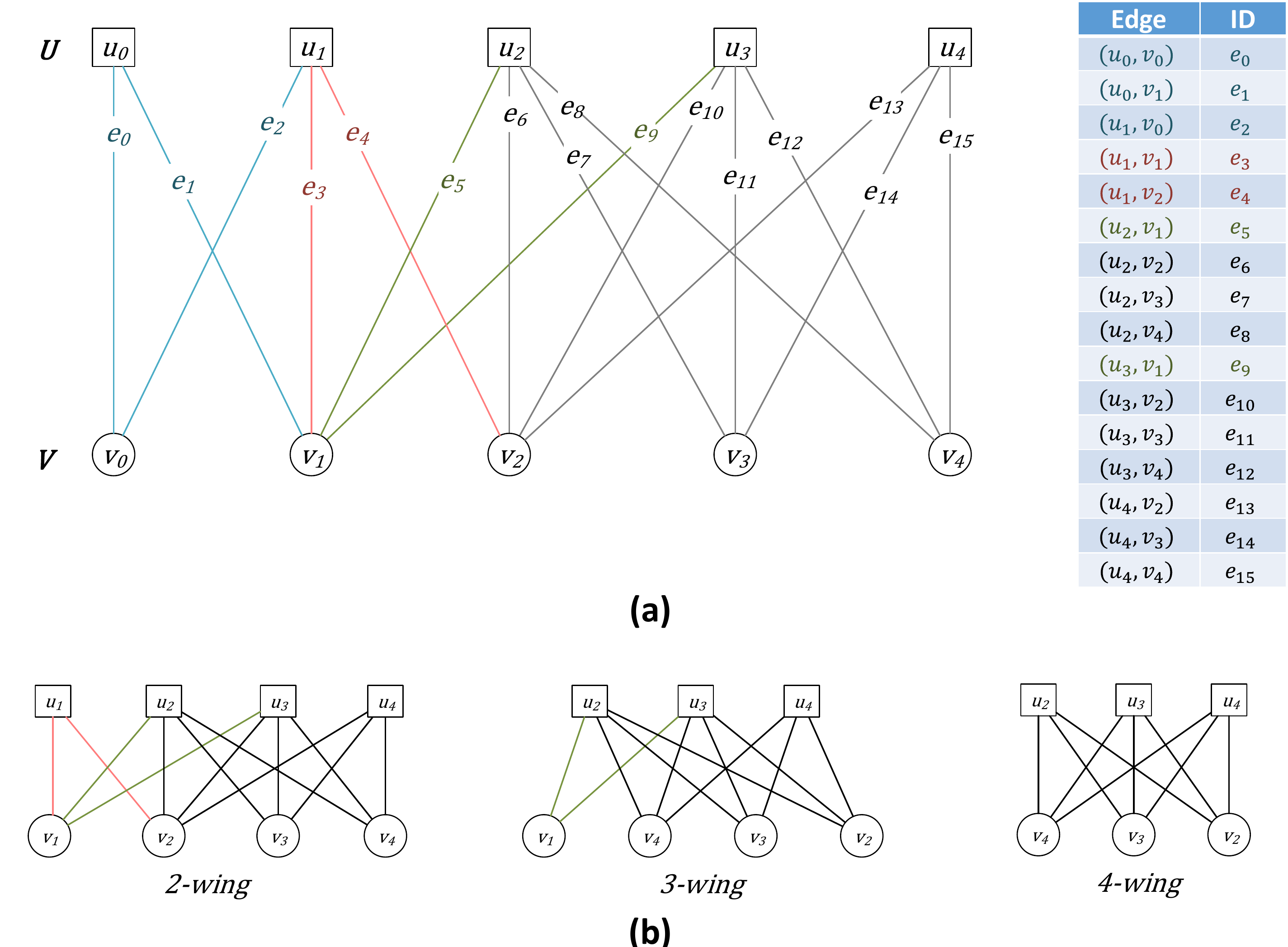}     
\caption{(a) Bipartite graph $G(U, V, E)$ which is also a $1$-wing. (b) Wing decomposition 
hierarchy of $G(U, V, E)$ -- edges colored blue, red, green and black have wing numbers of 1, 2, 3 and 4, respectively.}
    \label{fig:wingDemo}
\end{figure}

Existing algorithms for decomposing bipartite graphs typically employ an iterative bottom-up 
peeling approach~\cite{sariyucePeeling, shiParbutterfly}, wherein entities~(edges 
and vertices for wing and tip decomposition, respectively) with the minimum \emph{support} 
(butterfly count) are \emph{peeled} in each iteration. Peeling an entity $l$ involves deleting 
it from the graph and updating the support of other entities that share butterflies with $l$.
However, the huge number of butterflies in bipartite graphs
makes bottom-up peeling computationally demanding and 
renders large graphs intractable for 
decomposing by sequential algorithms. 
For example, \emph{trackers} -- a bipartite network of internet domains and the trackers 
contained in them, has $140$ million edges but more than $20$ trillion butterflies.

Parallel computing is widely used to scale such high complexity analytics to large
datasets~\cite{Park_2016,smith2017truss,10.1145/3299869.3319877}. However, the bottom-up peeling approach used in existing parallel frameworks~\cite{shiParbutterfly} severely
restricts parallelism by peeling entities in a strictly increasing order of their entity numbers~(wing or tip numbers). Consequently, it takes
a very large number of iterations to peel an entire graph, for example, 
it takes $>31$ million iterations to peel all edges of the \emph{trackers}
dataset using bottom-up peeling. 
Moreover, each peeling iteration is
\emph{sequentially dependent} on support updates in all 
prior iterations, thus mandating synchronization of parallel 
threads before each iteration.
Hence, the conventional approach of parallelizing workload within each 
iteration~\cite{shiParbutterfly} suffers from heavy thread synchronization, 
and poor parallel scalability. 

In this paper, we propose a novel two-phased peeling approach for
generalized bipartite graph decomposition. Both phases in the proposed approach exploit 
parallelism \emph{across multiple levels} of the decomposition 
hierarchy to 
drastically reduce the number of parallel peeling iterations and in turn, the amount of thread synchronization.  
The first phase creates a coarse hierarchy which divides 
the set of entity numbers into \textit{few non-overlapping} 
ranges. It accordingly partitions the entities by iteratively 
peeling the ones with support in the lowest range. 
A major implication of range-based partitioning is that each 
iteration peels a large number of entities corresponding to a wide range of entity numbers. 
This results in \textit{large parallel workload} per iteration
and \textit{little synchronization}.

The second phase concurrently processes multiple partitions to compute the exact entity numbers. Absence of overlap between 
corresponding entity number ranges enables every 
partition to be peeled \textit{independently} of others.  
By assigning each partition exclusively to a single thread, 
this phase achieves parallelism with \emph{no global \looseness=-1synchronization}. 

We implement the two-phased peeling as a part of our Parallel Bipartite Network peelinG~(PBNG) framework which adapts this approach for both wing and tip decomposition. PBNG further encapsulates novel
workload optimizations that exploit batched peeling of
numerous entities in the first phase to dramatically improve computational efficiency of decomposition. Overall, our contributions can be 
summarized as follows:
\begin{enumerate}[leftmargin=*]
    \itemsep0em
    \item We propose a novel two-phased peeling approach for bipartite graph decomposition,
    that parallelizes workload across different levels of decomposition hierarchy. 
    The proposed methodology is implemented in our PBNG framework which generalizes it for
    both vertex and edge peeling. To the best of our knowledge, this 
    is the first approach to utilize parallelism
    across the levels of both wing and tip decomposition
    hierarchies. 
    
    \item Using the proposed two-phased peeling, we achieve 
    a dramatic reduction in the number of parallel peeling
    iterations and in turn, the thread synchronization.
    As an example, wing decomposition of \textit{trackers} dataset in PBNG requires only $2034$ parallel 
    peeling iterations, which is four orders 
    of magnitude less than existing parallel algorithms. 
    
    \item We develop novel optimizations 
    that are highly effective 
    for the two-phased peeling approach
    and dramatically reduce the work done by PBNG. 
    As a result, PBNG traverses only $3.3$
    trillion wedges during tip decomposition of internet domains 
    in \emph{trackers} dataset, compared to $211.1$ trillion
    wedges traversed by the state-of-the-art.
\end{enumerate}

We empirically evaluate PBNG on several real-world bipartite
graphs and demonstrate its superior scalability compared to state-of-the-art. 
We show that PBNG significantly expands the limits of 
current practice 
by decomposing some of the largest publicly available datasets in few minutes/hours, 
that existing algorithms cannot 
decompose in multiple \looseness=-1days. 

In a previous work~\cite{lakhotia2020receipt}, we developed a two-phased 
algorithm for tip decomposition~(vertex peeling). 
This paper generalizes the two-phased approach for peeling any set of
entities within a bipartite graph. We further present non-trivial 
techniques to adopt the two-phased peeling for wing decomposition~(edge peeling), which is known to reveal better quality dense subgraphs than 
tip decomposition~\cite{sariyucePeeling}.

\section{Background}
In this section, we formally define the problem statement and review existing methods
for butterfly counting and bipartite graph decomposition. Note that counting
is used to initialize \textit{support} (running count of butterflies) of each 
vertex or edge before peeling, and also inspires some optimizations to 
improve efficiency of decomposition.

Table~\ref{table:notations} lists some notations used in this paper. 
For description of a general approach, we use the term \emph{entity} to 
denote a vertex (for tip decomposition) or an edge (for wing decomposition), and \emph{entity number} to denote tip
or wing number (sec.\ref{sec:bottomup}), respectively. Correspondingly, notations $\bowtie_l$ and $\theta_l$ denote the support and entity number 
of entity $l$.


\begin{table}[htbp]
\centering
\caption{Notations and their definition}
\label{table:notations}
\begin{tabular}{|c|l|}
\hline
$G(W=(U, V), E)$                                 & bipartite graph $G$ with disjoint vertex sets $U(G)$ and $V(G)$, and edges $E(G)$ \\ \hline
$n/m$                                           & no. of  vertices in $G$ i.e. $n=|W|$ / no. of edges in $G$ i.e. $m = |E|$                                                                                               \\ \hline
$\alpha$ & arboricity of $G$ \cite{chibaArboricity} \\ \hline
$N_u/d_u$                                            & neighbors of vertex $u$ / degree of vertex $u$                                                                                                                  \\ \hline
$\bowtie_u^G$ / $\bowtie_e^G$                         &
no. of butterflies in $G$ that contain vertex $u$ / edge $e$                                                         \\ \hline

$\bowtie_u$ / $\bowtie_e$                         & support (runing count of butterflies) of vertex $u$ / edge $e$                                                                                \\ \hline
$\bowtie_U$ / $\bowtie_E$ &  support vector of all vertices in set $U(G)$ / all edges in $E(G)$\\ \hline

$\theta_u$ / $\theta_e$                                       & tip number of vertex $u$ / wing number of edge $e$                                                                                                                    \\ \hline
$\theta^{max}_U$/$\theta^{max}_E$ & maximum tip number of vertices in $U(G)$ / maximum wing number of edges in $E(G)$ \\ \hline

$P$                                                & number of vertex/edge partitions created by PBNG                                                                                    \\ \hline
$T$                                                & number of threads                                                                                                                             \\ \hline
\end{tabular}
\end{table}

\subsection{Butterfly counting}\label{sec:counting}
A butterfly (2,2-bicliques/quadrangle) can be viewed as a combination of two wedges 
with common endpoints. For example, in fig.\ref{fig:wingDemo}a, both wedges $(e_0, e_1)$ and $(e_2, e_3)$ have end points $v_0$ and $v_1$, and form a butterfly.
A simple way to count butterflies is to explore all wedges and combine the ones with 
common end points. However, this is computationally inefficient with complexity $\mathcal{O}\left(\sum_{u\in U}\sum_{v\in N_u}d_v\right) = \mathcal{O}\left(\sum_{v\in V}d_v^2\right)$ (if we 
use vertices in $U$ as end points). 

Chiba and Nishizeki \cite{chibaArboricity} developed an efficient vertex-priority 
quadrangle counting algorithm in which starting from each vertex $u$, 
only those wedges are expanded where $u$ has the highest degree. 
It has a theoretical complexity of  $\mathcal{O}\left(\sum_{(u,v)\in E}\min{(d_u, d_v)}\right) = \mathcal{O}\left(\alpha \cdot m\right)$, which is state-of-the-art for 
butterfly counting.
Wang et al.\cite{wangButterfly} further propose a cache-efficient version of this algorithm
that traverses wedges such that the degree of the last vertex is greater than the that of 
the start and middle vertices (alg.\ref{alg:counting}, line 10). Thus, wedge explorations 
frequently end at a small set of high degree vertices that can be cached.

The vertex-priority algorithm can be easily parallelized by concurrently processing 
multiple start vertices~\cite{shiParbutterfly, wangButterfly}. 
In PBNG, we use the per-vertex and per-edge counting variants of the parallel 
algorithm~\cite{shiParbutterfly, wangButterfly}, as shown in alg.\ref{alg:counting}. To avoid conflicts,
each thread is provided an individual $n$-element $wedge\_count$ array (alg.\ref{alg:counting}, line 5) for wedge aggregation, 
and butterfly counts of entities are incremented using atomic operations. 

\begin{algorithm}[htbp]
	\caption{Counting per-vertex and per-edge butterflies (\texttt{pveBcnt})}
	\label{alg:counting}
	\begin{algorithmic}[1]
	    \Statex{\textbf{Input:} Bipartite Graph $G(W=(U, V), E)$} 
	    \Statex{\textbf{Output:} Butterfly counts -- $\bowtie_u$ for each $u\in W(G)$, and $\bowtie_e$ for each $e\in E(G)$}
        \State{$\bowtie_u\ \leftarrow 0$ for each $u\in W(G)$;\ \ \ 
        $\bowtie_e\ \leftarrow 0$ for each  $e\in E(G)$ }\Comment{\emph{Initialization}}

        \State{Relabel vertices $W(G)$ in decreasing order of degree}\Comment{\emph{Priority assignment}}
        \ParForEach{$u\in W(G)$}
        \State{Sort $N_u$ in increasing order of new labels}
        \EndParForEach
        \ParForEach{vertex $start\in W(G)$} 
            \State{Initialize hashmap $wedge\_count$ to all zeros}
            \State{Initialize an empty wedge set $nzw \leftarrow \{\phi\}$}
            \ForEach{vertex $mid\in N_{start}$}
                \ForEach{vertex $last\in N_{mid}$}
                    \If{$(last\geq mid)$ or $(last\geq start)$}{\ break}\EndIf
                    \State{$wedge\_count[last] \leftarrow wedge\_count[last] + 1$}
                    \State{$\ nzw \leftarrow nzw \cup (mid, last)$}
                \EndForEach
            \EndForEach
            \ForEach{$last$ such that $wedge\_count[last] > 1$} \Comment{\emph{\textbf{Per-vertex counting}}}
                \State{$bcnt \leftarrow {wedge\_count[last] \choose 2}$} 
                \State{$\bowtie_{start}\leftarrow\ \bowtie_{start} +\ bcnt$;\ \  $\bowtie_{last}\leftarrow\ \bowtie_{last} +\ bcnt$} 
                \ForEach{$(mid, last)\in nzw$} 
                    \State{$\bowtie_{mid}\leftarrow\ \bowtie_{mid} +\  (wedge\_count[last]-1)$}
                \EndForEach
            \EndForEach
            \ForEach{$(mid, last)\in nzw$ such that $wedge\_count[last] > 1$} \Comment{\emph{\textbf{Per-edge counting}}}
                \State{Let $e_1$ and $e_2$ denote edges $(start, mid)$ and $(mid, last)$, respectively}
                \State{$\bowtie_{e_1}\leftarrow\ \bowtie_{e_1} +\ (wedge\_count[last]-1)$;\ \  $\bowtie_{e_2}\leftarrow\ \bowtie_{e_2} +\ (wedge\_count[last]-1)$}
            \EndForEach            
            
        \EndParForEach
	\end{algorithmic}
\end{algorithm}

\subsection{Bipartite Graph Decomposition}\label{sec:bottomup}
Sariyuce et al.\cite{sariyucePeeling} introduced $k-$tips and $k-$wings as a butterfly dense vertex and edge-induced subgraphs, respectively. They
are formally defined as follows:

\begin{definition}\label{def:kwing}
A bipartite subgraph $H \subseteq G$, induced on edges $E(H)\subseteq E(G)$, is a \textbf{k-wing} iff
\begin{itemize}
    \itemsep0em
    \item each edge $e \in E(H)$ is contained in at least k butterflies,
    \item any two edges $(e_1,e_2)\in E'$ is connected by a series of butterflies,
    \item $H$ is maximal i.e. no other $k-$wing in $G$ subsumes $H$.
\end{itemize}
\end{definition}

\begin{definition}\label{def:ktip}
A bipartite subgraph $H\subseteq G$, induced on vertex sets $U(H)\subseteq U(G)$ and $V(H)=V(G)$, is a \textbf{k-tip} iff
\begin{itemize}
    \itemsep0em
    \item each vertex $u \in U(H)$ is contained in at least k butterflies,
    \item any two vertices $\{u,u'\}\in U(H)$ are connected by a series of butterflies,
    \item $H$ is maximal i.e. no other $k-$tip in $G$ subsumes $H$.
\end{itemize}
\end{definition}


Both $k-$wings and $k-$tips are hierarchical as a $k-$wing/$k-$tip
completely overlaps with a $k'-$wing/$k'-$tip for all $k'\leq k$. 
Therefore, instead of storing all $k-$wings, a \textit{wing number} 
$\theta_e$ of an edge $e$
is defined as the maximum $k$ for which 
$e$ is present in a $k-$wing. 
Similarly, \textit{tip number} $\theta_u$ of 
a vertex $u$ is the maximum $k$ for which 
$u$ is present in a $k-$tip.
Wing and tip numbers act as
a space-efficient indexing from which
any level of the $k-$wing and $k-$tip hierarchy, respectively, can be quickly 
retrieved~\cite{sariyucePeeling}. 
In this paper, we 
study the problem of finding wing and tip numbers, also known as 
\textit{wing and tip decomposition}, respectively.


Bottom-Up Peeling (\texttt{BUP}) is a commonly employed 
technique to compute wing 
decomposition~(alg.\ref{alg:bottomup}).
It initializes the support of each edge using per-edge butterfly counting~(alg.\ref{alg:bottomup}, line 1), 
and then iteratively peels the edges with 
minimum support until no edge remains. When an edge $e\in E$ is peeled, its support 
in that iteration is recorded as its wing number~(alg.\ref{alg:bottomup}, line 4).
Further, for every edge $e'$ that shares butterflies with $e$, 
the support $\bowtie_{e'}$ is decreased corresponding to the
removal of those butterflies. 
Thus, edges are peeled in a non-decreasing order of wing numbers. 

Bottom-up peeling for tip decomposition utilizes a similar
procedure for peeling vertices.
A \emph{crucial distinction} here is that in tip decomposition, vertices in only one
of the sets $U(G)$ or $V(G)$ are peeled as a $k-$tip consists of all vertices from 
the other set~(defn.\ref{def:ktip}). For clarity of description, we assume that $U(G)$ is 
the vertex set to peel.
As we will see later in sec.\ref{sec:tipPBNG}, this distinction renders the two-phased
approach of PBNG highly suitable for decomposition. 

Runtime of bottom-up peeling is dominated by wedge traversal required to find
butterflies that contain the entities being peeled (alg.\ref{alg:bottomup}, lines 7-9). 
The overall complexity for wing decomposition is $\mathcal{O}\left(\sum_{(u,v)\in E(G)}\sum_{v\in N_u}d_v\right) = \mathcal{O}\left(\sum_{u\in U(G)}\sum_{v\in N_u}d_ud_v\right)$. 
Relatively, tip decomposition has a lower complexity of $\mathcal{O}\left(\sum_{u\in U(G)}\sum_{v\in N_u}d_v\right) = \mathcal{O}\left(\sum_{v\in V(G)}d_v^2\right)$, which is 
still quadratic in vertex degrees and very high in absolute terms.

\begin{algorithm}[]
	\caption{Wing decomposition using bottom-up peeling (\texttt{BUP})}
	\label{alg:bottomup}
	\begin{algorithmic}[1]
	    \Statex{\textbf{Input:} Bipartite graph $G(W=(U, V), E)$} 
	    \Statex{\textbf{Output:} Wing numbers $\theta_e\ \forall\ e\in E(G)$}
        \State{$\bowtie_{E}\ \leftarrow$ \texttt{pveBcnt($G$)}}\Comment{\emph{Counting for support initialization (alg.\ref{alg:counting})}}
        \While{$E(G)\neq\ \{\phi \}$}\Comment{\emph{Peeling}}
            \State{$e\leftarrow \arg\min_{e\in E(G)}\{\bowtie_e \}$}
            \State{$\theta_e \leftarrow\  \bowtie_e,\ \ E(G)\leftarrow E(G)\setminus \{e\}$}
            \State{\texttt{\textsc{update}(}$e, \theta_e, \bowtie_{E}, G$\texttt{)}}    
        \EndWhile
        \Statex{}
        \Function{\texttt{update}}{$e=(u,v), \theta_e, \bowtie_{E}, G$}\label{func:update}
            \ForEach{$v'\in N_u \setminus \{v\}$}\Comment{\textit{Find butterflies}}
                \State{Let $e_1$ denote edge $(u, v')$}
                \ForEach{$u' \in N_{v'} \setminus \{u\}$ such that $(u', v)\in E(G)$} 
                    \State{Let $e_2$ and $e_3$ denote edges $(u', v)$ and $(u', v')$, respectively}
                    \State{$\bowtie_{e_1} \leftarrow \max\left(\theta_e,\ \bowtie_{e_1}-1\right)$;\ \  $\bowtie_{e_2} \leftarrow \max\left(\theta_e,\ \bowtie_{e_2}-1\right)$;\ \  $\bowtie_{e_3} \leftarrow \max\left(\theta_e,\ \bowtie_{e_3}-1\right)$}\Comment{\textit{Update support}}
                \EndForEach
            \EndForEach 
        \EndFunction
	\end{algorithmic}
\end{algorithm}

\subsection{Bloom-Edge-Index}\label{sec:beIndex}
Chiba and Nishizeki~\cite{chibaArboricity} proposed storing
wedges derived from the computational patterns of their butterfly counting algorithm, as a space-efficient
representation of all butterflies. Wang et al.\cite{wangBitruss}
used a similar representation termed 
\emph{Bloom-Edge-Index~(BE-Index)} for quick 
retrieval of butterflies containing peeled edges 
during wing decomposition. 
We extensively utilize BE-Index not just for computational efficiency, but also
for enabling parallelism in wing decomposition. In this subsection, we give a brief 
overview of some key concepts in this regard.

The butterfly counting algorithm
assigns priorities~(labels) to all vertices
in a decreasing order of their degree~(alg.\ref{alg:counting}, line 2). 
Based on these priorities, a structure called \emph{maximal priority $k-$bloom}, which is 
the basic building block of BE-Index, is
defined as follows~\cite{wangBitruss}:
\begin{definition}
A maximal priority $k-$bloom $B(W=(U, V), E)$ is a $(2,k)-$biclique 
(either $U(B)$ or $V(B)$ has exactly two vertices, each connected to all $k$ vertices in $V(B)$ or $U(B)$, respectively) that satisfies the following conditions:
\begin{enumerate}
    \item The highest \textbf{priority} vertex in $W(B)$ belongs to the 
    set ($U(B)$ or $V(B)$) which has exactly two vertices, and
    \item $B$ is \textbf{maximal} i.e. there exists no $(2,k)-$biclique $B'$ such that $B\subset B' \subseteq G$ and $B'$ satisfies condition $1$.
\end{enumerate}
\end{definition}

\paragraph{\underline{Maximal Priority Bloom Notations}:} The vertex set ($U(B)$ or $V(B)$) containing the 
highest priority vertex is called the \emph{dominant set} of $B$. Note that each vertex in the non-dominant 
set has exactly two incident edges in $E(B)$, that are said to be \emph{twins} of 
each other in bloom $B$. 
For example, in the graph $G'$ shown in fig.\ref{fig:beIndexDemo}, the subgraph induced on 
$\{v_0, v_1, u_0, u_1\}$ is a maximal priority $2-$bloom with $v_1\in V(B)$ as the 
highest priority vertex and twin edge pairs  $\{e_0, e_1\}$ and $\{e_2, e_3\}$.
The twin of an edge $e$ in bloom $B$ is denoted by $twin(e,B)$.
The cardinality of the non-dominant vertex set of bloom $B$ is called the \emph{bloom number} of $B$.
Wang et al.\cite{wangBitruss} further prove the following properties of maximal 
priority blooms:
\begin{property}\label{prop:bloomNum}
    A $k-$bloom $B(W=(U, V), E)$ consists of exactly 
    ${k\choose 2} = \frac{k*(k-1)}{2}$ butterflies. 
    Each edge $e\in E(B)$ is contained in exactly $k-1$
    butterflies in $B$.
    Further, edge $e$ shares all $k-1$ butterflies with $twin(e, B)$, and one butterfly each with all other edges $e'\in 
    E(B)\setminus\{twin(e, B)\}$.
\end{property}

\begin{property}\label{prop:unique}
    A butterfly in $G$ must be contained in exactly one maximal priority $k-$bloom.
\end{property}

    

Note that the butterflies containing an edge $e$, and the other edges in those butterflies, can be 
obtained by exploring all blooms that contain $e$. 
For quick access to blooms of an edge and vice-versa, BE-Index is defined as follows:

\begin{definition}
BE-Index of a graph $G(W=(U, V), E)$ is a bipartite graph $I(W=(U, V), E)$ that links
all maximal priority blooms in $G$ to the respective edges within the blooms.
\begin{itemize}
    \itemsep0em
    \item \textbf{W(I) --} Vertices in $U(I)$ and $V(I)$ uniquely 
    represent all maximal priority blooms in $G$ and edges in $E(G)$, 
    respectively. Each vertex $B\in U(I)$ also stores the bloom number 
    $k_B(I)$ of the corresponding bloom.

    \item \textbf{E(I) --} There exists an edge $(e, B)\in E(I)$ if and 
    only if the corresponding bloom $B\subseteq G$ contains the edge $e\in
    E(G)$. Each edge $(e, B)\in E(I)$
is labeled with $twin(e,B)$.
\end{itemize}
\end{definition}

\paragraph{\underline{BE-Index Notations}:} For ease of explanation, we refer to a maximal 
priority  bloom as simply bloom. We use the same notation $B$~(or $e$) to denote both a 
bloom~(or edge) and its representative vertex in BE-Index. Neighborhood of a 
vertex $B\in U(I)$ and $e\in V(I)$ is denoted by $N_B(I)$ and $N_e(I)$, respectively. The 
bloom number of $B$ in BE-Index $I$ is denoted by $k_B(I)$.
Note that $k_B(I) = \frac{\abs{N_B(I)}}{2}$.

\begin{figure}[htbp]
    \centering
\includegraphics[width=\linewidth]{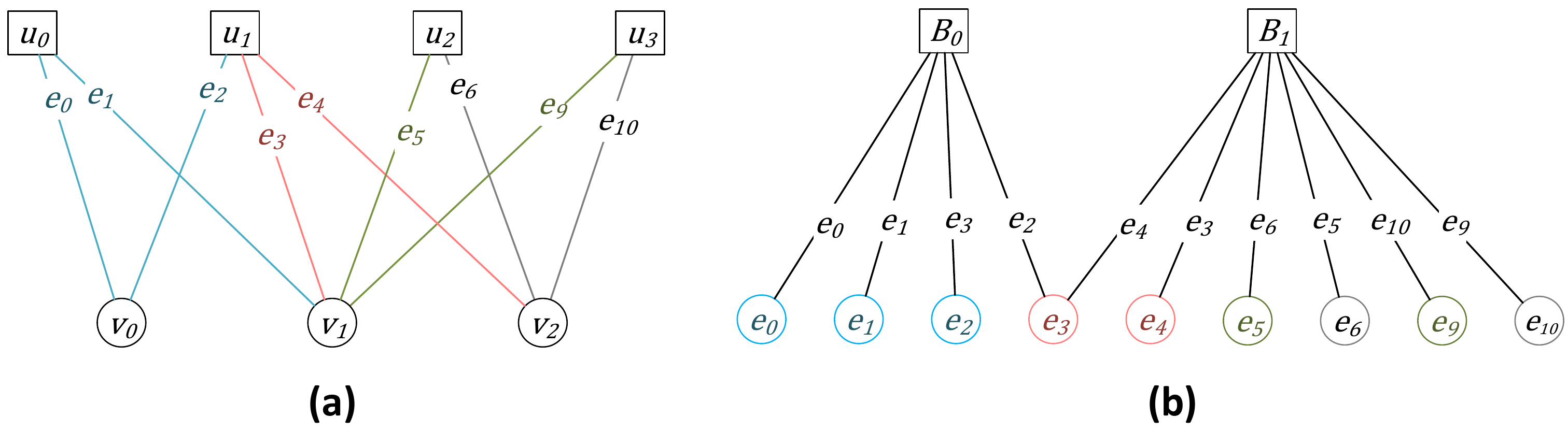}     
\caption{(a) Bipartite graph $G'(W=(U, V), E)$ (b) BE-Index of $G'$ with two maximal priority blooms.}
    \label{fig:beIndexDemo}
\end{figure}

Fig.\ref{fig:beIndexDemo} depicts a graph $G'$ (subgraph of $G$ from 
fig.\ref{fig:wingDemo})
and its BE-Index. $G'$ consists of two maximal priority blooms: 
(a) $B_0$ with dominant set $V(B_0)=\{v_0, v_1\}$ and 
$k_{B_0}(I)=2$, and (b) $B_1$ with dominant vertex set 
$V(B_1)=\{v_1, v_2\}$ and $k_{B_0}(I)=3$. 
As an example, edge $e_3$ is a part of $1$ butterfly in $B_0$ shared with 
twin $e_2$, and $2$ butterflies in $B_1$ shared with twin $e_4$. With all
other edges in $B_0$ and $B_1$, it shares one butterfly each. 

\paragraph{\underline{Construction of BE-Index}:} Index construction can be easily 
embedded within the counting procedure~(alg.\ref{alg:counting}). 
Each pair of endpoint vertices $\{start, last\}$ of wedges explored during 
counting, represents the dominating set of a bloom~(with $last$ as the highest priority vertex)
containing the edges $\{start, mid\}$ and $\{mid, last\}$ for all midpoints $mid$. 
Lastly, for a given vertex $mid$, edges $\{start, mid\}$ and 
$\{mid, last\}$ are twins of each other. Thus, the space
and computational complexities of BE-Index construction are bounded by the the wedges explored during counting which is $\mathcal{O}\left(\alpha \cdot m\right)$.

\paragraph{\underline{Wing Decomposition with BE-Index}:} Alg.\ref{alg:bePeel} depicts
the procedure to peel an edge $e$ using BE-Index $I$. Instead of traversing wedges in $G$ to find butterflies of $e$, 
edges that share butterflies with $e$ are found
by exploring 2-hop neighborhood of
$e$ in $I$~(alg.\ref{alg:bePeel}, line 7).
Number of butterflies shared with these edges in each bloom is also obtained analytically using property~\ref{prop:bloomNum}~(alg.\ref{alg:bePeel}, lines 4 and 8).
Remarkably, peeling an edge $e$ using alg.\ref{alg:bePeel} requires at most $\bowtie_{e}^G$ 
traversal in BE-Index~\cite{wangBitruss}.
Thus, it reduces the computational complexity of wing decomposition
to $\mathcal{O}\left(\sum_{e\in E(G)}\bowtie_e^G\right)$.
However, it is still proportional to the number of butterflies which can be
enormous for large graphs.

\begin{algorithm}[htbp]
	\caption{Support update during edge peeling, using BE-Index}
	\label{alg:bePeel}
	\begin{algorithmic}[1]
        \Function{\texttt{update}}{$e, \theta_e, \bowtie_{E}, \text{BE-Index}\ I(W=(U, V), E)$}\label{func:update}
            \ForEach{bloom $B\in N_e(I)$}
                \State{$e_t\leftarrow twin(e, B)$, $\ \ k_B(I) \leftarrow$ bloom number of $B$ in $I$}
                \State{$\bowtie_{e_t}\ \leftarrow\ \max\left(\theta_e,\ \bowtie_{e_t} -\ (k_B(I)-1)\right)$}
                \State{$E(I) \leftarrow E(I) \setminus \{(e, B), (e_t, B)\}$}
                \State{$k_B(I) \leftarrow k_B(I)-1$}\Comment{\textit{Update bloom number}}
                \ForEach{$e' \in N_{B}(I)$} 
                    \State{$\bowtie_{e'} \leftarrow \max\{\theta_e,\ \bowtie_{e'}-1\}$}\Comment{\textit{Update support}}
                \EndForEach
            \EndForEach 
        \EndFunction
	\end{algorithmic}
\end{algorithm}


\subsection{Challenges}\label{sec:challenges}
Bipartite graph decomposition is computationally very expensive and 
parallel computing is widely used to accelerate such workloads. However, 
state-of-the-art parallel framework P\textsc{ar}B\textsc{utterfly}~\cite{shiParbutterfly, julienne} 
is based on bottom-up peeling and only utilizes parallelism within each peeling iteration. 
This restricted parallelism is due to the following sequential 
dependency between iterations -- \textit{support updates in an 
iteration guide the choice of entities to peel in the subsequent iterations.} 
Hence, even though P\textsc{ar}B\textsc{utterfly} is work-efficient~\cite{shiParbutterfly}, its scalability is 
limited because:

\begin{enumerate}[leftmargin=*]
    \item It incurs large number of iterations and low parallel workload per iteration. Due to the resulting synchronization and load imbalance, intra-iteration parallelism is insufficient for substantial acceleration.\\
    \textbf{Objective 1} is therefore, to design a parallelism 
    aware peeling methodology for bipartite graphs that reduces synchronization and exposes large amount of parallel workload.

    \item It traverses an enormous amount of wedges~(or 
    bloom-edge links in BE-Index) to
    retrieve butterflies removed by peeling. This is 
    computationally expensive and can be infeasible on large datasets, 
    even for a parallel algorithm.\\
    \textbf{Objective 2} is therefore, to reduce the amount of traversal in practice. 
\end{enumerate}






\section{Parallel Bipartite Network peelinG (PBNG)}\label{sec:PBNG}
In this section, we describe a generic parallelism friendly two-phased peeling approach for bipartite 
graph decomposition~(targeting objective~$1$, sec.\ref{sec:challenges}). We further demonstrate how this approach
is adopted individually for tip and wing decomposition in our Parallel Bipartite Network peelinG~(PBNG) 
framework.
\subsection{Two-phased Peeling}\label{sec:twoPhase}
\begin{figure}[htbp]
    \centering
\includegraphics[width=0.92\linewidth]{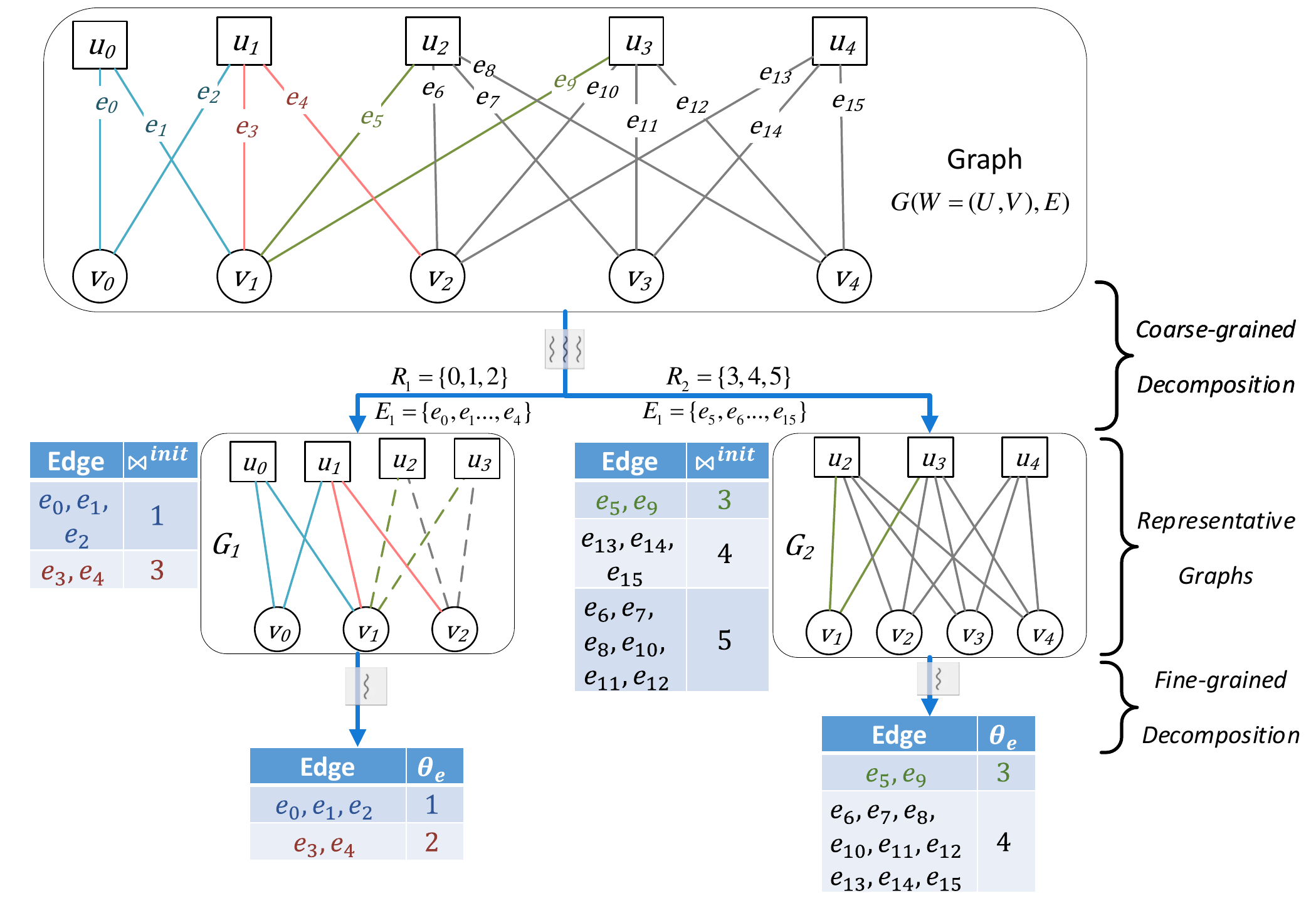}     
\caption{Graphical illustration of PBNG's two-phased peeling for wing decomposition of the graph $G$ from fig.\ref{fig:wingDemo}. The coarse-grained decomposition divides $E(G)$ into $P=2$ partitions using 
parallel peeling iterations. The fine-grained decomposition peels each partition using a single thread but concurrently processes multiple partitions.}
    \label{fig:example}
\end{figure}

The fundamental observation underlining our approach is that entity number $\theta_l$ for 
an entity $l$ only depends on the number of butterflies shared between $l$ and other entities
with entity numbers \textit{no less than} $\theta_l$. 
Therefore, given a graph $G(W=(U, V), E)$ and per-entity butterfly counts in $G$~(obtained from counting), 
only the \textit{cumulative} effect of peeling all entities with entity
number strictly smaller than $\theta_l$,
is relevant for computing $\theta_l-$level~($\theta_l-$tip or $\theta_l-$wing) in the decomposition hierarchy.
Due to commutativity of addition, the \emph{order} of peeling these entities has \emph{no impact} on
$\theta_l-$level.

This insight allows us to eliminate the constraint of deleting only 
minimum support entities in each iteration, which bottlenecks the available parallelism. 
To find $\theta_l-$level, all entities with entity number less than $\theta_l$ can be 
peeled concurrently, providing sufficient parallel workload. However, for 
every possible $\theta_l$, peeling all entities with smaller entity number will be computationally very inefficient.
To avoid this inefficiency, we develop a novel two-phased \looseness=-1approach.

\subsubsection{Coarse-grained Decomposition}\label{sec:coarse}
The first phase divides the spectrum of all possible entity 
numbers $[0, \theta^{max}]$ into $P$ smaller non-overlapping
ranges ${R_1, R_2\dots R_P}$, where $\theta^{max}$
is the maximum entity number in $G$, and $P$ is a 
user-specified parameter. A range $R_i$ represents a 
set of entity numbers $[\theta(i), \theta(i+1))$, 
such that $R_i \cap R_j = \{\phi\}$ for all $j\neq i$.
These ranges are computed using a heuristic described
in sec.\ref{sec:adaptive}.
Corresponding to each range $R_i$, PBNG also computes the partition $L_i$ comprising all entities
whose entity numbers lie in $R_i$. 
Thus, instead of finding the exact entity number 
$\theta_l$ of an entity $l$, the first phase of PBNG computes \textit{bounds} on $\theta_l$. 
Therefore, we refer to this phase as \emph{Coarse-grained Decomposition (PBNG CD)}. 
The absence of overlap between the ranges allows each subset to be peeled independently of others 
in the second phase, for exact entity number \looseness=-1computation.

Entity partitions are computed by iteratively peeling 
entities whose support lie in the minimum range~(alg.\ref{alg:cd},lines 5-13). 
For each partition, the first peeling iteration in PBNG CD scans all entities to find the peeling set, denoted
as $activeSet$~(alg.\ref{alg:cd}, line 9). In subsequent iterations,
$activeSet$ is computed jointly with support updates. Thus, unlike bottom-up peeling, 
PBNG CD does not require a priority queue data structure which makes support updates 
relatively \looseness=-1cheaper. 

PBNG CD can be visualized as a generalization of bottom-up peeling~(alg.\ref{alg:bottomup}). In each iteration, the latter peels
entities with minimum support~($\abs{R_i}=1$ for all $i$), whereas 
PBNG CD peels entities with support in a 
broad \textit{custom range}~($\abs{R_i}\geq1$). 
For example, in fig.\ref{fig:example}, 
PBNG CD divides edges $E(G)$
into two partitions corresponding to ranges $R_1=[0, 2]$ and $R_2=[3, 5]$, whereas bottom-up peeling 
will create $4$ partitions corresponding to every individual level in the decomposition 
hierarchy~($\theta=\{1, 2, 3, 4\}$). Setting $P\ll \theta^{max}$ ensures a large number of 
entities peeled per iteration~(\textit{sufficient parallel workload}) and significantly fewer 
iterations~(\textit{dramatically less synchronization}) compared to bottom-up \looseness=-1peeling.

In addition to the ranges and partitions, PBNG CD also computes a support initialization 
vector $\bowtie^{init}$. For an entity $l\in L_i$, 
$\bowtie^{init}_l$ is the number of butterflies that $l$  
shares \emph{only} with entities in partitions $L_j$ such that $j\geq i$. In other words, it represents
the aggregate effect of peeling entities with entity number in 
ranges lower than $R_i$. 
During iteative peeling in PBNG CD, this number is inherently 
generated after the last peeling iteration of $R_{i-1}$ and copied
into $\bowtie^{init}(l)$~(alg.\ref{alg:cd}, lines 6-7). 
For example, in fig.\ref{fig:example}, support of $e_5$ after peeling $E_1=\{e_0, e_1, e_2, e_3, e_4\}$ is $3$, which is recorded in $\bowtie^{init}_{e_5}$.

\begin{algorithm}[htbp]
	\caption{PBNG Coarse-grained Decomposition (PBNG CD) for wing decomposition}
	\label{alg:cd}
	\begin{algorithmic}[1]
	    \Statex{\textbf{Input:} Bipartite graph $G(W=(U, V), E)$, \# partitions $P$} 
	    \Statex{\textbf{Output:} Ranges $\{\theta(1), \theta(2)\dots \theta(P+1)\}$, Edge Partitions $\{E_1, E_2\dots E_P\}$, Support initialization vector $\bowtie^{init}_E$}
        \State{$E_i \leftarrow \{\phi\}\ \forall\ i\in \{1,2\dots P\}$}
        \State{Initial support $\{\bowtie_E\}\leftarrow$ \texttt{pveBcnt($G$)}}\Comment{Ref: alg.\ref{alg:counting}}
        \State{$I(W=(U, V), E)\leftarrow$ BE-Index of $G$}
        \State{$\theta(1)\leftarrow0$,  $\ i\leftarrow 1$, $\ tgt\leftarrow$ target butterflies (workload) per partition}
        \While{$\left(E(G)\neq\ \{\phi \}\right)$ and $\left(i\leq P\right)$}
            \ParForEach{$e\in E(G)$}\Comment{\textit{Support Initialization Vector}}
                \State{$\bowtie^{init}_e \leftarrow\ \bowtie_e$}
            \EndParForEach
            \State{$\theta(i+1)\leftarrow\text{\texttt{\textsc{find\_range}(}}E(G),\ \bowtie_{E},\ tgt\text{\texttt{)}}$}\Comment{\textit{Upper Bound}}
            \State{$activeSet\leftarrow$ all edges $e\in E(G)$ such that $\ \theta(i) \leq\ \bowtie_e\ < \theta(i+1)$}
            \While{$activeSet\neq \{\phi \}$}\Comment{\textit{Peel edges}}
                \State{$E_i \leftarrow E_i\cup activeSet$, $\ E(G) \leftarrow E(G)\setminus activeSet$}
                \State{\texttt{\textsc{parallel\_update(}$activeSet,\ \theta(i),\ \bowtie_E,\  I$\texttt{)}}} 
                \State{$activeSet\leftarrow$ all edges $e\in E(G)$ such that $\ \theta(i) \leq\ \bowtie_e\ < \theta(i+1)$}
            \EndWhile
            \State{$i\leftarrow i+1$}
        \EndWhile
        \Statex{}
        \Function{\texttt{find\_range}}{$E(G),\ \bowtie_E,\ tgt $}\label{func:findhi}
            \State{Initialize hashmap $work$ to all zeros}
            \ForEach{$\theta\ \in\  \bowtie_E$} 
                \State{$work[\theta] \leftarrow \sum_{e\in E(G)}\left(\bowtie_e\cdot \mathds{1}(\bowtie_e \leq\ \theta)\right)$}
            \EndForEach
            \State{$\theta_{ub}\leftarrow$ $\arg\min\left(\theta\right)$ such that $work[\theta] \geq tgt$}
            \State{return $\theta_{ub}+1$}
            
        \EndFunction
        \Statex{}
        \Function{\texttt{parallel\_update}}{$activeSet,\ \theta_e,\ \bowtie_{E},\ \text{BE-Index}\ I(W=(U, V), E)$}\label{func:update}
            \State{Initialize hashmap $count$ to all zeros}
            \ParForEach{edge $e\in activeSet$}
                \ForEach{bloom $B\in N_e(I)$}\Comment{\textit{Update support atomically}}
                    \State{$e_t\leftarrow twin(e, B)$, $\ \ k_B(I) \leftarrow$ bloom number of $B$ in $I$}
                    \If{$\left(e_t\notin activeSet\right)$\ or\ $\left(edgeID(e_t) < edgeID(e)\right)$}
                        \State{$\bowtie_{e_t}\ \leftarrow\ \max\left(\theta_e,\ \bowtie_{e_t} -\ (k_B(I)-1)\right)$}
                        \State{$E(I) \leftarrow E(I) \setminus \{(e, B), (e_t, B)\}$}
                        \State{$count[B] \leftarrow count[B] + 1$}
                        \ForEach{$e' \in N_{B}(I)\ $ such that $\ twin(e', B)\notin activeSet$}
                            \State{$\bowtie_{e'} \leftarrow \max\{\theta_e,\ \bowtie_{e'}-1\}$}
                        \EndForEach
                    \EndIf
                \EndForEach 
            \EndParForEach
            \ParForEach{$B\in U(I)$ such that $count[B]>0$}
                \State{$k_B(I) \leftarrow k_B(I) - count[B]$}\Comment{\textit{Update bloom number}}
            \EndParForEach
        \EndFunction
	\end{algorithmic}
\end{algorithm}

\subsubsection{Fine-grained Decomposition}\label{sec:fine}
The second phase computes exact entity numbers and is called \emph{Fine-grained Decomposition (PBNG FD)}.
The key idea behind PBNG FD is that if we have the knowledge of all butterflies that each entity 
$l\in L_i$ shares \emph{only} with entities in partitions $L_j$ such that $j\geq i$, $L_i$ 
can be peeled independently of all other partitions.
The $\bowtie^{init}$ vector computed in PBNG CD precisely indicates the count of such 
butterflies~(sec.\ref{sec:coarse}) and hence, is used to \emph{initialize support values} in PBNG FD.
PBNG FD exploits the resulting independence among partitions
to concurrently process multiple partitions using 
sequential bottom up peeling. 
Setting $P\gg T$ ensures that PBNG FD can be efficiently 
parallelized across partitions on $T$ threads. 
Overall, both phases in PBNG circumvent strict sequential dependencies between different 
levels of decomposition hierarchy to efficiently parallelize the peeling process.

The two-phased approach can potentially double the 
computation required for peeling. However, we note that 
since partitions are peeled independently in PBNG FD, 
support updates are not communicated across the partitions. 
Therefore, to improve computational efficiency, 
PBNG FD operates on a smaller representative subgraph $G_i\subseteq G$
for each partitions $L_i$.
Specifically, $G_i$ preserves a butterfly $\bowtie$ iff
it satisfies both of the following \looseness=-1conditions:
\begin{enumerate}[leftmargin=*]
    \item $\bowtie$ contains multiple entities within $L_i$.
    \item $\bowtie$ only contains entities from partitions
    $L_j$ such that $j\geq i$. If $\bowtie$
    contains an entity from lower ranged partitions, then it does not exist in $\theta(i)$-level 
    of decomposition hierarchy~(lowest entity number in $L_i$). Moreover, the impact of removing $\bowtie$ on 
    the support of entities in $L_i$, is already accounted for in $\bowtie^{init}$~(sec.\ref{sec:coarse}).
\end{enumerate}
For example, in fig.\ref{fig:example}, $G_1$ contains the butterfly 
$\bowtie\ =\left(u_1, v_1, u_2, v_2\right)$ because (a) it contains multiple edges $\{e_3, e_4\}\in L_1$ and satisfies condition $1$, and (b) all edges in $\bowtie$ are from $L_1$ or $L_2$ and hence, 
it satisfies condition $2$. However, $G_2$ does not contain this butterfly because two if its edges 
are in $L_1$ and hence, it does not satisfy condition $2$ for $G_2$.

\subsubsection{Range Partitioning}\label{sec:adaptive}
In PBNG CD, the first step for computing a partition $L_i$ is to find the range 
$R_i=[\theta(i), \theta(i+1))$~(alg.\ref{alg:cd}, line 8). 
For load balancing, $\theta(i+1)$ should be computed\footnote{$\theta(i)$ is directly
obtained from upper bound of previous range $R_{i-1}$.} 
such that the all partitions $L_i$ pose uniform workload in PBNG FD.
However, the representative subgraphs and the corresponding workloads are not known prior to 
actual partitioning. 
Furthermore, exact entity numbers are not known either and hence, we cannot 
determine beforehand, exactly which entities will lie in $L_i$ for different values of $\theta(i+1)$. 
Considering these challenges, PBNG uses two \emph{proxies} for range determination: 
\begin{enumerate}[leftmargin=*]
    \item\emph{Proxy 1 $\rightarrow$} current support $\bowtie_l$ of an entity $l$ is used as a proxy for its entity number.
    \item \emph{Proxy 2 $\rightarrow$} complexity of peeling individual entities in $G$ is 
    used as a proxy to estimate peeling workload in representative subgraphs.
\end{enumerate}

Now, the problem is to compute $\theta(i+1)$ such that \emph{estimated} workload of $L_i$ as 
per proxies, is close to the average workload per partition denoted as $tgt$.
To this purpose, PBNG CD creates a bin for each support value, 
and computes the aggregate workload of entities in that bin.
For a given $\theta(i+1)$, estimated workload of peeling $L_i$ is the sum of workload of all bins corresponding to support less than\footnote{All entities with entity numbers less than $\theta(i)$ are already peeled before PBNG CD computes $R_i$.} $\theta(i+1)$. 
Thus, the workload of $L_i$ as a function of $\theta(i+1)$ can be 
computed by a prefix scan of individual bin workloads~(alg.\ref{alg:cd}, lines 17-18). 
Using this function, the upper bound is chosen such that the 
estimated workload of $L_i$ is close to but no less than $tgt$~(alg.\ref{alg:cd}, line 19).

\paragraph{Adaptive Range Computation:} Since range determination uses \emph{current support} 
as a proxy for entity numbers, the target workload for each partition $L_i$ is covered by the entities added to $L_i$ in its very first peeling iteration in PBNG CD. 
After the support updates in this iteration, more entities may be
added to $L_i$ and final workload estimate of $L_i$ may significantly exceed $tgt$. 
This can result in significant load imbalance among the partitions and
potentially, PBNG CD could finish in much fewer than $P$ partitions. To avoid this scenario, we implement the following two-way adaptive range determination:
\begin{enumerate}[leftmargin=*]
    \itemsep0em
    \item Instead of statically computing an average target, we 
    dynamically update $tgt$ for every partition based on the remaining workload and the number of partitions to create. If a partition gets too much workload, the 
    target for subsequent partitions is automatically reduced, thereby preventing a situation where 
    all entities get peeled in $\ll P$ \looseness=-1partitions.
    \item A partition $L_i$ likely covers many more entities
    than the initial estimate based on proxy 1. 
    The second adaptation scales down the dynamic target
    for $L_{i}$ in an attempt to bring the actual workload
    close to the intended value. It assumes predictive local
    behavior i.e. $L_i$ will overshoot the target similar
    to $L_{i-1}$. Therefore, the scaling factor is computed
    as the ratio of initial workload estimate of $L_{i-1}$ 
    during $\theta(i)$ computation, and final estimate based 
    on all entities in $L_{i-1}$.
\end{enumerate}

\subsubsection{Partition scheduling in PBNG FD}\label{sec:schedule}
While adaptive range determination(sec.\ref{sec:adaptive}) tries to create partitions with
uniform estimated workload, the actual workload per partition in PBNG FD depends on the 
the representative subgraphs $G_i$ and can still have significant variance. Therefore, 
to \textit{improve load balance} across threads, we use scheduling strategies inspired from Longest 
Processing Time~(LPT) scheduling rule which is a well known $\frac{4}{3}$-approximation algorithm \cite{graham1969bounds}. 
We use the workload of $L_i$ as an indicator of its execution time in the following runtime scheduling 
mechanism:
\begin{itemize}[leftmargin=*]
    \item \textit{Dynamic task allocation$\rightarrow$} 
    All partition IDs are inserted in a task queue. 
    When a thread becomes idle, it pops a unique ID from the queue and processes the corresponding partition. Thus, all 
    threads are busy until every partition is \looseness=-1scheduled.
    \item \textit{Workload-aware Scheduling$\rightarrow$} Partition IDs in the task queue are sorted in a
    decreasing order of their workload. Thus, partitions with highest workload get scheduled first 
    and the threads processing them naturally receive fewer tasks in the future.
    Fig.\ref{fig:schedule} shows how workload-aware scheduling can improve the efficiency of dynamic 
    allocation.
\end{itemize}

\begin{figure}[htbp]
    \centering
\includegraphics[width=0.75\linewidth]{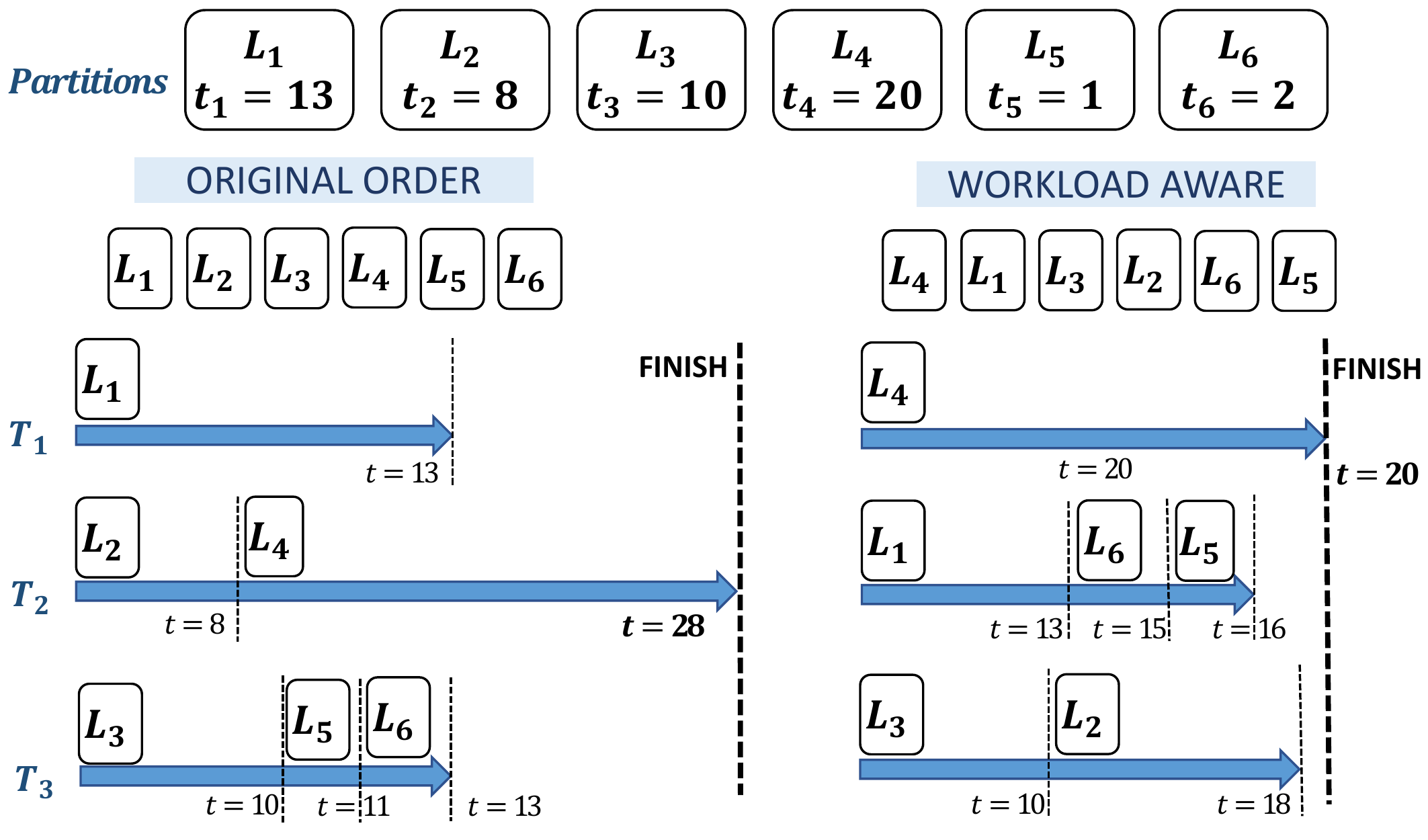}     
\caption{Benefits of Workload-aware Scheduling (WaS) in a $3$-thread ($T_1,T_2$ and $T_3$) system. Top row shows entity partitions with estimated time to peel them in PBNG FD. Dynamic allocation without WaS finishes in $28$ units of time compared to $20$ units with \looseness=-1WaS.}
    \label{fig:schedule}
\end{figure}



\subsection{Tip Decomposition}\label{sec:tipPBNG}
In this section, we give a brief overview of PBNG's two-phased peeling~(sec.\ref{sec:twoPhase}) applied 
for tip decomposition. A detailed description of the same 
is provided in our previous work~\cite{lakhotia2020receipt}.

For tip decomposition, PBNG CD divides the vertex set $U(G)$ into $P$ partitions -- $\{U_1, U_2, \dots, U_P\}$. 
Peeling a vertex $u\in U(G)$ requires traversal of all wedges
with $u$ as one of the endpoints. Therefore, range 
determination in PBNG CD uses wedge count of vertices in 
$G$, as a proxy to estimate the workload of peeling each partition $U_i$. Moreover, since only one of the vertex sets
of $G$ is peeled, at most two vertices of a butterfly can be a part of the 
peeling set~($activeSet$). Hence, support updates to a vertex $u'$ from 
different vertices in $activeSet$ correspond to disjoint butterflies in $G$.
The net update to support $\bowtie_{u'}$ can be simply computed by atomically aggregating the 
updates from individual vertices in \looseness=-1$activeSet$.

PBNG FD also utilizes the fact that any butterfly $\bowtie$ contains at most two vertices $\{u,u'\}$ 
in the vertex set $U(G)$ being peeled~(sec.\ref{sec:bottomup}). 
If $u\in U_i$  and $u'\in U_j$, the two 
conditions for preserving $\bowtie$ in either representative graphs $G_i$ or $G_j$
are satisfied only when $i = j$~(sec.\ref{sec:fine}). 
Based on this insight, we construct $G_i$ as the \emph{subgraph induced} on vertices $W_i=(U_i, V(G))$. 
Clearly, $G_i$ preserves every butterfly $\left(u, v, u', v'\right)$ where $\{u, u'\}\in U_i$.
For task scheduling in PBNG FD~(sec.\ref{sec:schedule}), we use the total wedges in $G_i$ with endpoints in $U_i$ as an indicator of the workload of peeling $U_i$. 

Given the bipartite nature of graph $G$, any edge $(u, v)\in E(G)$ exists in exactly one of the 
subgraphs $G_i$ and thus, the collective space requirement of all induced subgraphs is bounded by
$O(m)$. Moreover, by the design of representative (induced) subgraphs, PBNG FD for tip decomposition 
traverses only those wedges for which both the endpoints are in the same partition. This dramatically reduces the amount 
of work done in PBNG FD compared to bottom-up peeling and PBNG CD.
Note that we do not use BE-Index for tip decomposition due to the following reasons:
\begin{itemize}[leftmargin=*]
    \item Butterflies between two vertices are quadratic in the number of wedges between them, and
    wedge traversal (not butterfly count) determines the work done in tip decomposition. 
    Since BE-Index facilitates per-edge butterfly retrieval, peeling a vertex using BE-Index will
    require processing each of its edge individually and can result in increased computation if $\sum_{v\in V(G)}d_v^2 \ll \sum_{e\in E(G)}\bowtie_e^G$~(sec.\ref{sec:beIndex}).
    
    \item BE-Index has a \emph{high space complexity} of $\mathcal{O}\left(\sum_{(u,v)\in E(G)}\min{\left(d_v, 
    d_u\right)}\right)$ compared to just $\mathcal{O}\left(m\right)$ space needed to store $G$ and all 
    induced subgraphs $G_i$. This can make BE-Index based peeling infeasible even on machines
    with large amount of main memory. For example, BE-Index 
    of a user-group dataset \emph{Orkut}~($327$ million edges) has $26$ billion blooms, $150$ billion bloom-edge
    links and consumes $2.6$ TB memory.
\end{itemize}



\subsection{Wing Decomposition}
\subsubsection{Challenges}
Each butterfly consists of $4$ edges in $E(G)$ which 
is the entity set to decompose in wing decomposition. This is unlike tip decomposition
where each butterfly has only $2$ vertices from the decomposition set $U(G)$, and
results in the following issues:
\begin{enumerate}[leftmargin=*]
    
    \item When a butterfly $\bowtie$ is removed due to 
    peeling, the support of unpeeled edge(s) in $\bowtie$ should be reduced by exactly $1$ corresponding to this
    removal. 
    However, when multiple~(but not all) edges in $\bowtie$
    are concurrently peeled in the same iteration of PBNG CD, multiple updates with aggregate value $>1$ may be generated to unpeeled edges in $\bowtie$.

    \item It is possible that a butterfly $\bowtie$ 
    contains multiple but not all edges from a partition. 
    Thus, $\bowtie$ may need to be preserved in the representative subgraph of a partition, but will not be present in its edge-induced \looseness=-1subgraph. 
\end{enumerate}
Due to these reasons, a trivial extension of tip 
decomposition algorithm~(sec.\ref{sec:tipPBNG}) is not suitable for wing decomposition. 
In this section, we explore novel BE-Index based strategies to enable
two-phased peeling for wing decomposition.

\subsubsection{PBNG CD}
This phase divides the edges $E(G)$ into $P$ partitions -- $\{E_1, E_2, \dots, E_P\}$, as shown in alg.\ref{alg:cd}. 
Not only do we utilize BE-Index for computationally efficient support update computation in PBNG CD, we also utilize it to \emph{avoid conflicts} in parallel peeling iterations of PBNG CD. 
Since a butterfly is contained in exactly one maximal priority bloom~(sec.\ref{sec:beIndex}, property~\ref{prop:unique}), 
correctness of support updates within each bloom $B$ implies overall correctness of support updates 
in an iteration. To avoid conflicts, we therefore employ the following \emph{resolution
mechanism} for each bloom $B$:
\begin{enumerate}[leftmargin=*]
    \item If an edge $e$ and its twin $e'=twin(e, B)$ are both peeled in the same iteration, 
    then only the edge with highest index among $e$ and $e'$ updates (a)~the support of other edges in 
    $B$, and (b)~the bloom number $k_B(I)$~(alg.\ref{alg:cd}, lines 26-31). This is because all 
    butterflies in $B$ that contain $e$ also contain $e'$~(sec.\ref{sec:beIndex}, property\ref{prop:bloomNum}).
    
    \item If in an iteration, an edge $e\in activeSet$ but $e'=twin(e, B)\notin activeSet$, then the support $\bowtie_{e'}$ 
    is decreased by exactly ${k_B(I)-1}$ when $e$ is peeled. Other edges in $activeSet$
    do not propagate any updates to $\bowtie_{e'}$ via bloom $B$~(alg.\ref{alg:cd}, 
    lines 26-30). 
    This is because $e'$ is contained in exactly ${k_B(I)-1}$ butterflies in $B$, all of which are removed 
    when $e$ is peeled.
    To ensure that $k_B(I)$ correctly represents the butterflies shared between twin edges, support 
    updates from all peeled edges are computed prior to updating $k_B(I)$.
\end{enumerate}

Peeling an edge $e\in E(G)$ requires $\mathcal{O}\left(\bowtie_e\right)$ 
traversal in the BE-Index. 
Therefore, range determination in PBNG CD uses 
edge support as a proxy to estimate the workload of 
peeling each partition $E_i$. 

\begin{algorithm}[htbp]
	\caption{PBNG Fine-grained Decomposition (PBNG FD) for wing decomposition}
	\label{alg:fd}
	\begin{algorithmic}[1]
	    \Statex{\textbf{Input:} Graph $G(W=(U, V), E)$, BE-Index $I(W=(U, V), E)$, edge partitions $\{E_1, \dots, E_P\}$}
	    \Statex{\ \ \ \ \ \ \ \ \ \ Support initialization vector $\bowtie^{init}_E$}
	    \Statex{\textbf{Output:} Wing number $\theta_e$ for all $e\in E(G)$}
        
        \State{$\bowtie_{e}\  \leftarrow\  \bowtie^{init}_{e}$ for all $e\in E_i$}       \Comment{\textit{Initialize Support}}
        \State{$\{I_1, I_2, \dots, I_P\}\leftarrow$ \texttt{partition\_BE\_Index($G, I$)}}\Comment{\textit{BE-Indices for all partitions }}
        \State{Insert integers $\{1, 2, \dots, P\}$ in queue $Q$}
        \State{Compute $work[i] \leftarrow \sum_{e\in E_i}{\bowtie_e^{init}\ }$ for all $i\in Q$, and sort $Q$ by $work$}\Comment{\textit{LPT Scheduling}}
        \ParFor{$thread\_id \in \{1,2\dots T\}$}
            \While{$Q$ is not empty}
                \State{Atomically pop integer $i$ from $Q$}\Comment{\textit{Dynamic Task Allocation}}
                \While{$E_i \neq \{\phi\}$}\Comment{\textit{Peel partition}}
                    \State{$e \leftarrow \arg\min_{e\in E_i}\{\bowtie_e \}$}    \State{$\theta_e\leftarrow\ \bowtie_e$, $\ E_i\leftarrow E_i\setminus \{e\}$}
                    \State{\texttt{\textsc{update(}}$e,\theta_e,\bowtie_{E_i}, I_i$\texttt{)}}\Comment{\textit{Ref: alg.\ref{alg:bePeel}}}
                \EndWhile
            \EndWhile
        \EndParFor
        \Statex{}
        \Function{\texttt{partition\_BE\_Index}}{$G(W=(U, V), E), I(W=(U, V), E)$}\Comment{\textit{Compute partitions' BE-Indices}}
            \State{Let $p(e)$ denote the partition index of an edge $e$ i.e. $e\in E_{p(e)}$}
            \State{$V(I_i) \leftarrow E_i,\ U(I_i)\leftarrow \{\phi\},\ E(I_i) \leftarrow \{\phi\}\ \forall\ i\in \{1,2\dots P\}$} 
            \ParForEach{$B\in U(I)$}
                \State{Initialize hash map $wedge\_count$ to all zeros}
                \ForEach{$e\in N_B(I)$}
                    \State{$e_t \leftarrow twin(e, B)$}
                    \If{$p(e_t) \geq p(e)$}
                        \State{$U(I_i) \leftarrow U(I_i) \cup \{B\},\ \ E(I_i) \leftarrow E(I_i) \cup \left(e, B\right)$} \Comment{\textit{Add bloom-edge links}} 
                        \If{$\left(p(e_t)) > p(e)\right)$ or $\left(edgeID(e_t) < edgeID(e)\right) $}\Comment{\textit{Initialize bloom numbers}}
                            \State{$wedge\_count(I_i)\leftarrow wedge\_count(I_i) + 1$}
                        \EndIf
                    \EndIf
                \EndForEach
                \ForEach{partition index $i$ such that $B\in U(I_i)$}
                    \State{$k_B(I_i)\leftarrow \sum_{j \geq i}{wedge\_count(I_j)}$}\Comment{\textit{Compute bloom numbers~(prefix scan)}}
                \EndForEach
            \EndParForEach
            \State{return $\{I_1, I_2, \dots, I_P\}$}
            
        \EndFunction
	\end{algorithmic}
\end{algorithm}

\subsubsection{PBNG FD}
The first step for peeling a partition $E_i$ in PBNG FD, is to construct the corresponding BE-Index $I_i$
for its representative subgraph $G_i$. One way to do so is to compute $G_i$ and then 
use the index construction algorithm~(sec.\ref{sec:beIndex}) to construct $I_i$. However, this approach 
has the following drawbacks:
\begin{itemize}[leftmargin=*]
    \item Computing $G_i$ requires mining all edges in $E(G)$ that share butterflies with edges $e\in 
    E_i$, which can be \emph{computationally expensive}. 
    Additionally, the overhead of index construction even for a few hundred partitions can be significantly \looseness=-1large.
    \item Any edge $e\in E_i$ can potentially exist in all subgraphs $G_j$ such that $j\leq i$. 
    Therefore, creating and storing all subgraphs $G_i$ requires $\mathcal{O}(mP)$ \emph{memory space}.
\end{itemize}

To avoid these drawbacks, we directly compute $I_i$ for each $E_i$ by
partitioning the BE-Index $I$ of the original graph $G$~(alg.\ref{alg:fd}, lines 12-25).
Our partitioning mechanism ensures that all butterflies satisfying the two preservation conditions~(sec.\ref{sec:fine}) for a partition $E_i$, are represented in its BE-Index $I_i$.

Firstly, for an edge $e\in E_i$, its link $(e, B)$ with a bloom $B$ is preserved in $I_i$ if and only if
the twin $e_t=twin(e,B)\in E_j$ such that $j\geq i$~(alg.\ref{alg:fd}, lines 19-20). 
Since all butterflies in $B$ that contain $e$ also contain $e_t$~(sec.\ref{sec:beIndex}, property~\ref{prop:bloomNum}), 
none of them need to be preserved in $I_i$ if $j<i$. Moreover, if $\{e, e_t\}\in E_i$, their 
contribution to the bloom number $k_B(I_i)$ is counted only once~(alg.\ref{alg:fd}, lines \looseness=-121-22).

Secondly, for a space-efficient representation, $I_i$ does not store a link $\left(e', B\right)$
if $e'\notin E_i$. 
However, while such an edge $e'$ will not participate in peeling of $E_i$, it may be a part of a butterfly 
in $B$ that satisfies both preservation conditions for $I_i$~(sec.\ref{sec:fine}). 
For example, fig.\ref{fig:beIndexDemo} shows the representative subgraph $G_1$ and BE-Index for the 
partition $E_1$ generated by PBNG CD in fig.\ref{fig:example}. 
For space efficiency, we do not store the links $\{(e_5, B_1), (e_6, B_1), (e_9, B_1), (e_{10}, B_1)\}$ 
in $I_1$. 
However, the two butterflies in $B_1$ -- $(e_3, e_4, e_5, e_6)$ and $(e_3, e_4, e_9, e_{10})$,
satisfy both preservation conditions for $I_1$, and may be needed when peeling $e_3$ or $e_4$.
In order to account for such butterflies, we \emph{adjust the bloom number} $k_B(I_i)$
to correctly represent the number of butterflies in $B$ that only contain edges from \emph{all} 
partitions $E_j$ such that $j\geq i$~(alg.\ref{alg:fd}, lines 23-24). For example, in 
fig.\ref{fig:beIndexDemo}b, we initialize the bloom number 
of $B_1$ to $k_{B_1}(I_1) = 3$ even though $\abs{N_{B_1}(I_1)}=2$. Thus, $k_{B_1}(I_1)$ correctly 
represents the ${3\choose 2}=3$ butterflies in $B_1$, 
that contain edges only from $\cup_{j\geq 1}\{E_j\}$
\footnote{This is unlike the BE-Index $I$ of graph $G$ where bloom numb
er $k_B(I)=\frac{\abs{N_B(I)}}{2}$~(sec.\ref{sec:beIndex}).}. 

After the BE-Indices for partitions are computed, PBNG FD dynamically schedules partitions 
on threads, where they are processed using sequential 
bottom-up peeling. Here, the aggregate initial support of a partition's edges (given by $\bowtie^{init}_E$ vector) is used
as an indicator of its workload~(alg.\ref{alg:fd}, line 4) for LPT scheduling~~(sec.\ref{sec:fine}).

\section{Analysis}\label{sec:analysis}

\subsection{Correctness of Decomposition Output}\label{sec:correctness}
In this section, we prove the correctness of wing numbers and tip numbers computed by PBNG. 
We will exploit comparisons with sequential \texttt{BUP}~(alg.\ref{alg:bottomup}) 
and hence, first establish the follwing lemmas:

\begin{lemma}\label{lemma:independence}
In \texttt{BUP}, the support $\bowtie_e$ of an edge $e$ at any time $t$ before the first edge with wing number $\geq\theta_e$ is peeled, depends on the cumulative effect of all edges peeled till $t$ and is independent of the 
order in which they are peeled.
\end{lemma}
\begin{proof}
Let $\bowtie_e^G$ denote the number of butterflies in $G$ that contain $e$, $S$ be the set of 
vertices peeled till time $t$ and $\bowtie_{e, e_a, e_b, e_c}$ denote a butterfly containing $e$ and other edges $\{e_a, e_b, e_c\}$ such that $a>b>c$~(for uniqueness of \looseness=-1representation).
If $e_a$, $e_b$ or $e_c$ are peeled till $t$,
corresponding to the removal of $\bowtie_{e, e_a, e_b, e_c}$, only one of them~(the first edge to be peeled) 
reduces the support $\bowtie_e$ by a unit.
Since \texttt{BUP} peels edges in a non-decreasing order of their wing numbers, 
$\theta_{e'}<\theta_{e}$ for all $e'\in S$.
Hence, current support~(alg.\ref{alg:bottomup}, line 11) of 
$e$ is given by $\bowtie_e\ =\ \bowtie_e^G - \sum_{\bowtie_{e, e_a, e_b, e_c}}{\mathds{1}(e_a\in S\ \texttt{or}\ e_b\in S\ 
\texttt{or}\ e_c\in S)}$. The first term of RHS is constant for a given $G$, and by commutativity of addition, 
the second term is independent of the order in which contribution of individual butterflies are added. Therefore, 
$\bowtie_e$ is independent of the order of peeling in $S$.
\end{proof}

\begin{lemma}\label{lemma:parallel}
Given a set $S$ of edges peeled in an iteration, the parallel peeling in PBNG CD (alg.\ref{alg:cd}, lines 21-33) correctly updates the support of remaining edges in $E(G)$.
\end{lemma}
\begin{proof}
Parallel updates are correct if for an edge $e$ not yet assigned to any partition, 
$\bowtie_e$ decreases by exactly $1$ for each butterfly of $e$ deleted in a peeling iteration.
By property~\ref{prop:unique}~(sec.\ref{sec:beIndex}), parallel updates are correct if this holds true for 
butterflies deleted within each bloom.

Consider a bloom $B(W=(U, V), E)$ and let $\bowtie_{e, e_t, e', e'_t}\subseteq B$
be a butterfly containing edges $\{e, e_t, e', e'_t\}$, where $e_t=twin(e, B)$ and $e'_t=twin(e', B)$~(property~\ref{prop:bloomNum}, sec.\ref{sec:beIndex}), such 
that $\bowtie_{e, e_t, e', e'_t}$ is deleted in the $j^{th}$ peeling iteration of
PBNG CD. Clearly, neither of $e$, $e_t$, $e'$ or $e'_t$ have been peeled before $j^{th}$ iteration. We analyze updates to support $\bowtie_e$ corresponding to the deletion of butterfly $\bowtie_{e, e_t, e', e'_t}$. 
If $S$ denotes the set of edges peeled in $j^{th}$ iteration, $\bowtie$ will be deleted if and only if either of 
the following is \looseness=-1true:
\begin{enumerate}[leftmargin=*]
    \item $e\in S\rightarrow$ In this case, $e$ is already assigned to a partition and updates to $\bowtie_e$
    have no impact on the output of PBNG.

    \item $e\notin S$ but $e_t\in S \rightarrow$ $\bowtie_e$ is reduced by exactly $(k_B-1)$ via bloom $B$, which amounts to unit reduction per each butterfly in $B$ that contains $e$~(property~\ref{prop:bloomNum}, sec.\ref{sec:beIndex}). No  updates from other edges are propagated to $\bowtie_e$ via $B$~(alg.\ref{alg:cd}, line 30).
    
    \item $\{e, e_t\}\notin S$ but $e'\in S$ or $e'_t\in S$ $\rightarrow$ $\bowtie_e$ is decreased by exactly $1$
    when $e'/e'_t$ is peeled. If both $\{e', e'_t\}\in S$, then $\bowtie_e$ is decreased by exactly $1$
    when the highest indexed edge among $e'$ and $e'_t$ is peeled~(alg.\ref{alg:cd}, line 26).
    Both the links $(e', B)$ and $(e'_t, B)$ are deleted and bloom number $k_B$ is decreased by $1$. 
    Hence, subsequent peeling of any edges will not generate support updates corresponding to $\bowtie_{e, e_t, e', e'_t}$. 
\end{enumerate}
Thus, the removal of $\bowtie$ during peeling decreases the support of $e$ by exactly $1$ if $e$ is not scheduled for peeling yet. 
\end{proof}

Lemmas~\ref{lemma:independence} and \ref{lemma:parallel} together show that whether we peel a set $S\subseteq E(G)$ in its original order in \texttt{BUP}, or in parallel in any order in PBNG CD, the support of edges with wing numbers higher than all edges in $S$ would be the same after peeling $S$. 
Next, we show that PBNG CD correctly computes the edge partitions corresponding to wing number ranges
and finally prove that PBNG FD outputs correct wing numbers for all edges. For ease of explanation, we use $\bowtie_e(j)$ to denote the support of a edge $e$ after $j^{th}$ peeling iteration in PBNG CD.

\begin{lemma}\label{lemma:cd1}
There cannot exist an edge $e$ such that $e\in E_i$ and $\theta_e \geq \theta(i+1)$.
\end{lemma}
\begin{proof}
Let $j$ be the first iteration in PBNG CD that wrongly peels a set of edges $S_{w}$, and assigns them to $E_i$ 
even though $\theta_e \geq \theta(i+1)\ \forall\ e\in S_w$. Let $S_{hi} \supseteq S_{w}$ be the set of all edges 
with wing numbers $\geq \theta(i+1)$ and $S_{lo}$ be the set of edges peeled till iteration $j-1$.  
Since all edges till $j-1$ iterations have been correctly peeled, $\theta_{e} < \theta(i+1)\ \forall\ e\in S_{lo}$.
Hence, $S_{hi} \subseteq E(G)\setminus S_{lo}$. 

Consider an edge $e\in S_{w}$. Since $e$ is peeled in $j^{th}$ iteration, ${\bowtie_e(j-1)} < \theta(i+1)$. From lemma \ref{lemma:parallel}, $\bowtie_e(j-1)$ is equal to the number of butterflies containing $e$ and edges only from $E\setminus S_{lo}$. Since $S_{hi} \subseteq E(G)\setminus S_{lo}$, there are at most $\bowtie_e(j-1)$ butterflies that
contain $e$ and other edges only from $S_{hi}$. By definition of wing number~(sec.\ref{sec:bottomup}), $\theta_e \leq \bowtie_e(j-1) < \theta(i+1)$, which is a contradiction. Thus, no such $e$ exists, $S_{w}=\{\phi\}$, and all edges 
in $E_i$ have tip numbers less than $\theta(i+1)$.
\end{proof}

\begin{lemma}\label{lemma:cd2}
There cannot exist an edge $e$ such that $e\in E_i$ and $\theta_e < \theta(i)$.
\end{lemma}
\begin{proof}
Let $R_i$ be the smallest range for which there exists an edge $e$ such that $\theta(i)\leq \theta_e < \theta(i+1)$, but $e\in E_{i'}$ for some $i'>i$.
Let $j$ be the last peeling iteration in PBNG CD that assigns edges to $E_i$, and $S_{hi}=\cup_{i'>i}E_{i'}$ denote the set of edges not peeled 
till this iteration.
Clearly, $e\in S_{hi}$ and for each edge $e'\in S_{hi}$, $\bowtie_{e'}(j) \geq \theta(i+1)$, otherwise $e'$ would be peeled in or before the $j^{th}$ 
iteration. 
From lemma \ref{lemma:parallel}, all edges in $S_{hi}$~(including $e$) participate in at least $\theta(i+1)$ butterflies that contain edges only from $S_{hi}$.
Therefore, $e$ is a part of $\theta(i+1)-wing$ (def.\ref{def:kwing})
and by the definition of wing number, $\theta_e \geq \theta(i+1)$, 
which is a contradiction. 
\end{proof}

\begin{theorem}\label{theorem:cd}
PBNG CD (alg.\ref{alg:cd}) correctly computes the edge partitions corresponding to every wing number range.
\end{theorem}
\begin{proof}
Follows directly from lemmas~\ref{lemma:cd1} and \ref{lemma:cd2}.
\end{proof}

\begin{theorem}\label{theorem:pbngCorrect}
\text{PBNG} correctly computes the wing numbers for all $e\in E(G)$.
\end{theorem}
\begin{proof}
From theorem~\ref{theorem:cd}, $\theta(i)\leq \theta_e< \theta(i+1)$ for
each edge $e\in E_i$.
Consider an edge partition $E_i$ and Let $S_{lo}=E_1\cup E_2\dots E_{i-1}$ denote the set of edges peeled before $E_i$ in PBNG CD. 
From theorem~\ref{theorem:cd}, $\theta_e\geq \theta(i)$ for all edges $e\in E_i$, and $\theta_{e}<\theta(i)$ for all edges $e\in S_{lo}$. 
Hence, $S_{lo}$ will be peeled before $E_i$ in \texttt{BUP} as well. 
Similarly, any edge in $S_{hi}=E_{i+1}\cup E_{i+2}\dots\cup E_{P}$ 
will be peeled after $E_i$ in \texttt{BUP}. Hence, support updates
to any edge $e\in S_{hi}$ have no impact on wing numbers computed for 
edges in $E_i$.

For each edge $e\in E_i$, PBNG FD initializes $\bowtie_e$ using $\bowtie^{init}_e$, which is the support of $e$ in PBNG CD 
just after $S_{lo}$ is peeled. From lemma \ref{lemma:parallel}, 
this is equal to the support of $e$ in \texttt{BUP} just after
$S_{lo}$ is peeled. 

Note that both PBNG FD and \texttt{BUP} employ same algorithm~(sequential bottom-up peeling) to peel edges in $E_i$. Hence, 
to prove the correctness of wing 
numbers generated by PBNG FD, it suffices to show that when an edge in $E_i$ is peeled, support updates propagated to other edges in $E_i$ via each bloom $B$, are the same in PBNG FD and \texttt{BUP}.
\begin{enumerate}[leftmargin=*]
	\item Firstly, every bloom-edge link $(e, B)$ where $e\in E_i$ and $twin(e, B)\in E_i\cup S_{hi}$, is preserved in the partition's
	BE-Index $I_i$. 
	Thus, when an edge $e\in E_i$ is peeled, the set of affected edges in $E_i$ are correctly retrieved from $I_i$~(same as the set retrieved
	from BE-Index $I$ in \texttt{BUP}).
	
	\item Secondly, by construction~(alg.\ref{alg:fd}, lines 21-22 and 23-24), the initial bloom number of $B$ in $I_i$ is equal to the number of those twin edge pairs in $B$, for which
	both edges are in $E_i$ or higher ranged partitions~(not
	necessarily in same partition). Thus, $k_B(I_i)=\frac{\sum_{e\in E(B)}{\mathds{1}\left(\theta_e\geq \theta(i)\ \texttt{and}\ \theta_{twin(e, B)}\geq \theta_i\right)}}{2}$, which in turn is the bloom number $k_B(I)$ 
	in \texttt{BUP} just before $E_i$ is peeled.
\end{enumerate}
Since both PBNG FD and \texttt{BUP} have identical support values
just before peeling $E_i$, the updates computed for edges in 
$E_i$ and the order of edges peeled in $E_i$ will be the same as well. Hence, the final wing numbers computed
by PBNG FD will be equal to those computed by \texttt{BUP}.
\end{proof}

The correctness of tip decomposition in PBNG can be proven in a similar fashion. A detailed derivation
for the same is given in theorem~$2$ of \cite{lakhotia2020receipt},

\subsection{Computation and Space Complexity}

To feasibly decompose large datasets, it is important for an algorithm to
be efficient in terms of computation and memory requirements. 
The following theorems show that for a reasonable upper bound 
on partitions $P$, 
PBNG is at least as computationally efficient as the best sequential decomposition algorithm \texttt{BUP}.

\begin{theorem}\label{theorem:wingComplexity}
For $P=\mathcal{O}\left(\frac{\sum_{e\in E(G)}{\bowtie_{e}^G}}{m}\right)$ partitions, wing decomposition in 
PBNG is work-efficient with computational complexity of 
$\mathcal{O}\left(\alpha\cdot m + \sum_{e\in E(G)}{\bowtie_{e}^G}\right)$,
where $\bowtie_{e}^G$ is the number of butterflies in $G$ that contain $e$.
\end{theorem}
\begin{proof}
PBNG CD initializes the edge support using butterfly counting algorithm with ${\mathcal{O}\left(\alpha\cdot m\right)}$ 
complexity. Since $\theta^{max}_E\leq m$, binning for range computation of each partition
can be done using an $\mathcal{O}(m)$-element array, such that $i^{th}$
element in the array corresponds to workload of edges
with support $i$~(alg.\ref{alg:cd}, lines 16-19). 
A prefix scan of the array gives the range 
workload as a function of the upper bound. 
Parallel implementations of binning and prefix scan perform $\mathcal{O}(m)$ computations per partition, amounting to 
$\mathcal{O}\left(P\cdot m\right)$ computations in entire PBNG CD. 
Constructing $activeSet$ for first peeling iteration of each partition requires an $\mathcal{O}(m)$ 
complexity parallel filter on remaining edges. Subsequent iterations construct $activeSet$ by tracking 
support updates doing $\mathcal{O}(1)$ work per update. 
Further, peeling an edge $e$ generates $\mathcal{O}(\bowtie_e^G)$ updates,
each of which can be applied in constant time using BE-Index~(sec.\ref{sec:beIndex}).
Therefore, total complexity of PBNG CD is $\mathcal{O}\left(\sum_{e\in E(G)}{\bowtie_{e}^G} + (P+\alpha)\cdot m\right)$.
	
PBNG FD partitions BE-Index $I$ to create individual BE-Indices for edge partitions. The partitioning~(alg.\ref{alg:fd}, lines 12-25)
requires constant number of traversals of entire set $E(I)$ and hence, 
has a complexity of $\mathcal{O}(\alpha\cdot m)$.
Each butterfly in $G$ is represented in at most one partition's BE-Index. Therefore, PBNG FD will also generate $\mathcal{O}\left(\sum_{e\in E(G)}{\bowtie_{e}^G}\right)$
support updates. 
Hence, the work complexity of PBNG FD is $\mathcal{O}\left(\sum_{e\in E(G)}{\bowtie_{e}^G} + \alpha\cdot m\right)$~(sec.\ref{sec:beIndex}).

The total work done by PBNG is $\mathcal{O}\left(\sum_{e\in E(G)}{\bowtie_{e}^G} + (P+\alpha)\cdot m\right)=\mathcal{O}\left(\sum_{e\in E(G)}{\bowtie_{e}^G} + \alpha\cdot m\right)$ if $P=\mathcal{O}\left(\frac{\sum_{e\in E(G)}{\bowtie_{e}^G}}{m}\right)$, which is the best-known time complexity of \texttt{BUP}. 
Hence, PBNG's wing decomposition is work-efficient.
\end{proof}

\begin{theorem}\label{theorem:tipComplexity}
For $P=\mathcal{O}\left(\frac{\sum_{u\in U}\sum_{v\in N_u}{d_v}}{n\log{n}}\right)$ partitions, tip decomposition 
in PBNG is work-efficient with computational complexity of $\mathcal{O}\left(\sum_{u\in U}\sum_{v\in N_u}{d_v}\right)$.
\end{theorem}
\begin{proof}
The proof is simliar to that of theorem~\ref{theorem:wingComplexity}.
The key difference from wing decomposition is that maximum tip number can be cubic in the size of the vertex set.
Therefore, range determination uses a hashmap with support values as the
keys. The aggregate workloads of the bins need to be sorted on keys
before computing the prefix scan. Hence, total complexity of range determination for tip decomposition in PBNG CD is 
$\mathcal{O}\left(P\cdot n\log{n}\right)$.
A detailed derivation is given in theorem~$3$ of \cite{lakhotia2020receipt}.
\end{proof}

Next, we prove that PBNG's space consumption is almost similar to the best known sequential algorithms.
\begin{theorem}
Wing decomposition in PBNG parallelized over $T$ threads consumes $\mathcal{O}\left(\alpha\cdot m + n\cdot T\right)$ memory space.
\end{theorem}
\begin{proof}
For butterfly counting, each thread uses a private $\mathcal{O}(n)$ array to accumulate wedge counts, resulting
to $\mathcal{O}(n\cdot T)$ space consumption~\cite{wangButterfly}. For peeling:
\begin{enumerate}
    \item PBNG CD uses the BE-Index $I(W=(U, V), E)$ whose space complexity is $\mathcal{O}\left(\alpha\cdot m\right)$. 
    \item PBNG FD uses individual BE-Indices for each partition. Any bloom-edge link $(e, B)\in E(I)$
    exists in at most one partition's BE-Index. Therefore, cumulative 
    space required to store all partitions'
    BE-Indices is 
    $\mathcal{O}\left(\alpha\cdot m\right)$.
\end{enumerate}
Thus, overall space complexity of PBNG's wing decomposition is $\mathcal{O}\left(\alpha\cdot m + n\cdot T\right)$.
\end{proof}

\begin{theorem}
Tip decomposition in PBNG parallelized over $T$ threads 
consumes $\mathcal{O}\left( m + n\cdot T\right)$ memory space.
\end{theorem}
\begin{proof}
For butterfly counting and peeling in PBNG, each thread uses a private $\mathcal{O}(n)$ array to accumulate wedge 
counts, resulting in $\mathcal{O}(n\cdot T)$ space consumption. 
In PBNG FD, any edge $e\in E(G)$ exists in the induced subgraph of at most 
one partition~(because partitions of $U(G)$ are disjoint). Hence, 
$\mathcal{O}\left( n + m\right)$ space is required to store $G$ and all 
induced subgraphs, resulting in overall $\mathcal{O}\left(m + n\cdot T\right)$ space complexity.
\end{proof}

\section{Optimizations}\label{sec:optimizations}
Despite the use of parallel computing resources, PBNG may consume a lot of time to decompose large graphs such as the \textit{trackers} dataset, that contain several trillion wedges~(for tip 
decomposition) or butterflies~(for wedge decomposition). 
In this section, we propose novel optimization techniques based on the two-phased peeling of PBNG,
that dramatically improve computational efficiency and make it feasible to decompose datasets
like \textit{trackers} in few minutes. 

\subsection{Batch Processing}\label{sec:batch}
Due to the broad range of entity numbers peeled in each iteration of PBNG 
CD~(sec.\ref{sec:coarse}), some iterations may peel a large number of entities.
Peeling individual entities in such iterations requires a large amount of traversal in $G$ or BE-Index $I$. 
However, visualizing such iterations as peeling a \emph{single large set} of 
entities can enable batch optimizations that drastically reduce the required 
computation.\vspace{-0.5em}

\paragraph{Tip Decomposition}
Here, we exploit the fact that (per-vertex) butterfly counting is computationally efficient and
parallelizble~(sec.\ref{sec:counting}). Given a vertex set $activeSet$, the 
number of wedges traversed for peeling $activeSet$ is given by $\wedge(activeSet) = \sum_{u\in activeSet}{\sum_{v\in N_u} d_v}$. 
However, number of wedges traversed for re-counting butterflies for remaining vertices is upper bounded by $\wedge_{cnt} = \sum_{(u,v)\in E}\min\left(d_u, d_v\right)$, which is constant for a given $G$~(sec.\ref{sec:counting}). 
If $\wedge(activeSet)>\wedge_{cnt}$, we \textit{re-compute butterflies} for all 
remaining vertices in $U$ instead of peeling $activeSet$. Thus, computational complexity of
a peeling iteration in PBNG tip decomposition is $\mathcal{O}\left(\alpha\cdot m\right)$.\vspace{-0.5em}

\begin{algorithm}[htbp]
	\caption{Batch computation of support updates from peeling a set of edges $activeSet$}
	\label{alg:batchPeel}
	\begin{algorithmic}[1]
        \Function{\texttt{update}}{$activeSet, \theta_e, \bowtie_{E}, \text{BE-Index}\ I(W=(U, V), E)$}\label{func:update}
        \State{Initialize hashmap $count$ to all zeros}
        \ParForEach{edge $e\in activeSet$}
            \ForEach{bloom $B\in N_e(I)$}
                \State{$e_t \leftarrow twin(e, B),\ \ k_B(I) \leftarrow$ bloom number of $B$ in $I$} 
                \If{$\left(e_t\notin activeSet\right)$\ or\ $\left(edgeID(e_t) < edgeID(e)\right)$}
                    \State{$\bowtie_{e_t}\ \leftarrow\ \max\left(\theta_e,\ \bowtie_{e_t} -\ k_B(I)\right)$}
                    \State{$count[B] \leftarrow count[B]+1$}\Comment{\textit{Aggregate updates at blooms atomically}}
                    \State{$E(I) \leftarrow E(I) \setminus \{(e, B), (e_t, B)\}$}
                    \EndIf
                \EndForEach 
            \EndParForEach
            \ParForEach{bloom $B\in U(I)$ such that $count[B]>0$}
                \State{$k_B(I) \leftarrow k_B(I) - count[B]$}\Comment{\textit{Update bloom number}}
                \ForEach{edge $e'\in N_B(I)$}
                    \State{$\bowtie_{e'} \leftarrow \max\{\theta_e,\ \bowtie_{e'}-count[B]\}$}\Comment{\textit{Update support atomically}}
                \EndForEach 
            \EndParForEach
        \EndFunction
	\end{algorithmic}
\end{algorithm}
\paragraph{Wing Decomposition}
For peeling a large set of edges $activeSet$, we use the batch processing proposed 
in~\cite{wangBitruss}. The key idea is that when an edge $e$ is peeled, the affected edges are 
discovered by exploring the neighborhood of
blooms in $N_e(I)$~(alg.\ref{alg:bePeel}). Therefore, 
the support updates from all edges in $activeSet$ can be 
\emph{aggregated at the blooms}~(alg.\ref{alg:batchPeel}, line 8), 
and then applied via a single traversal of their 
neighborhoods~(alg.\ref{alg:batchPeel}, lines 10-13). 
Thus, computational complexity of a peeling iteration in PBNG wing decomposition is
bounded by the size of BE-Index which is $\mathcal{O}\left(\alpha\cdot m\right)$.
While batch processing using BE-Index was proposed in \cite{wangBitruss}, we note that 
it is significantly more beneficial for PBNG CD compared to bottom-up peeling, 
due to the large number of edges peeled per iteration.

\subsection{Dynamic Graph Updates}\label{sec:dynamic}
After a vertex $u$ (or an edge $e$) is peeled in PBNG, it is excluded from future computation in the 
respective phase. However, due to the undirected nature of the graph, the adjacency list data structure for $G$~(or BE-Index $I$) still contains edges of $u$ (or bloom-edge links of $e$) that are interleaved 
with other edges.
Consequently, wedges incident on $u$ (or bloom-edge links of $e$), though not used in computation, 
are still explored even after $u$ (or $e$) is peeled. 
To prevent such wasteful exploration, we \textit{update the data structures} to remove edges incident on peeled vertices (or bloom-edge links of peeled edges). 

These updates can be performed jointly with the traversal required for
peeling. In tip decomposition, updating vertex support requires traversing
adjacencies of the neighbors of peeled vertices. Edges to peeled vertices can
be removed while traversing neighbors' adjacency lists.
In wing decomposition, updating edge support requires iterating over the affected blooms~(alg.\ref{alg:batchPeel}, lines 10-13) and their
neighborhoods $N_B(I)$. The bloom-edge links incident on peeled edges and their twins can be removed
from $N_B(I)$ during such traversal.

\section{Experiments}
In this section, we present detailed experimental results of 
PBNG for both tip and wing decomposition. In sec.\ref{sec:setup}, we list the datasets and describe
the baselines used for comparison. Secondly, in sec.\ref{sec:resultsWing}, we
provide a thorough evaluation of wing decomposition in PBNG. 
We (a)~compare PBNG against the baselines on 
several metrics, (b)~report empirical benefits of 
optimizations proposed in sec.\ref{sec:optimizations}, (c)~compare workload and execution time of coarse and fine
decomposition phases, and (d)~evaluate parallel scalability
of PBNG. Lastly, in sec.\ref{sec:resultsTip}, we report a similar evaluation of tip decomposition.

\subsection{Setup}\label{sec:setup}
We conduct the experiments on a 36 core dual-socket linux server with two Intel Xeon E5-2695 v4 processors@ 2.1GHz and 1TB DRAM. 
All algorithms are implemented in C++-14 and are compiled using G++ 9.1.0 with the -O3 optimization flag, and OpenMP v4.5 for 
multithreading.\vspace{1mm}

\noindent\textbf{\textit{Datasets}}: We use twelve unweighted bipartite graphs obtained from the KOBLENZ collection~\cite{konect} and Network 
Repository~\cite{rossi2015network}, whose characteristics are shown in table~\ref{table:datasets}.
To the best of our knowledge, these are some of the largest publicly available real-world bipartite datasets. 

\begin{table}[htbp]
\centering
\caption{Bipartite graphs with the corresponding number of butterflies ($\bowtie_G$ in billions), 
maximum tip numbers $\theta^{max}_U$~(for peeling $U$) and $\theta^{max}_V$~(for peeling $V$), and 
maximum wing number $\theta^{max}_E$.}
\label{table:datasets}
\resizebox{\linewidth}{!}{%
\begin{tabular}{|c|c|c|c|c|c|c|c|c|}
\hline
\textbf{Dataset} & \textbf{Description}                        & \textbf{$\mathbf{\abs{U}}$} & $\mathbf{\abs{V}}$ & $\mathbf{\abs{E}}$ & $\mathbf{\boldsymbol{\bowtie}_G}$(in B) & $\mathbf{\boldsymbol{\theta}^{max}_U}$ & $\mathbf{\boldsymbol{\theta}^{max}_V}$ & $\mathbf{\boldsymbol{\theta}^{max}_E}$ \\ \hline
Di-af            & Artists and labels affiliation from Discogs & 1,754,824                   & 270,772            & 5,302,276          & 3.3                                     & 120,101,751                            & 87,016,404                             & 15,498                                 \\ \hline
De-ti            & URLs and tags from www.delicious.com        & 4,006,817                   & 577,524            & 14,976,403         & 22.9                                    & 565,413                                & 409,807,620                            & 26,895                                 \\ \hline
Fr               & Pages and editors from French Wikipedia     & 62,890                      & 94,307             & 2,494,939          & 34.1                                    & 1,561,397                              & 645,790,738                            & 54,743                                 \\ \hline
Di-st            & Artists and  release styles from Discogs    & 1,617,944                   & 384                & 5,740,842          & 77.4                                    & 736,089                                & 1,828,291,442                          & 52,015                                 \\ \hline
It               & Pages and editors from Italian Wikipedia    & 2,255,875                   & 137,693            & 12,644,802         & 298                                     & 1,555,462                              & 5,328,302,365                          & 166,785                                \\ \hline
Digg             & Users and stories from Digg                 & 872,623                     & 12,472             & 22,624,727         & 1,580.5                                 & 47,596,665                             & 3,725,895,816                          &                  166,826                      \\ \hline
En               & Pages and editors from English Wikipedia    & 21,504,191                  & 3,819,691          & 122,075,170        & 2,036                                   & 37,217,466                             & 96,241,348,356                         & 438,728                                \\ \hline
Lj               & Users' group memberships in Livejournal     & 3,201,203                   & 7,489,073          & 112,307,385        & 3,297                                   & 4,670,317                              & 82,785,273,931                         & 456,791                                \\ \hline
Gtr              & Documents and words from Gottron-trec       & 556,078                     & 1,173,226          & 83,629,405         & 19,438                                  &  205,399,233                              &  38,283,508,375                           & 563,244                                \\ \hline
Tr               & Internet domains and trackers in them       & 27,665,730                  & 12,756,244         & 140,613,762        & 20,068                                  & 18,667,660,476                         & 3,030,765,085,153                      & 2,462,017                              \\ \hline
Or               & Users' group memberships in Orkut           & 2,783,196                   & 8,730,857          & 327,037,487        & 22,131                                  & 88,812,453                             & 29,285,249,823                         & -                                      \\ \hline
De-ut            & Users and tags from www.delicious.com       & 4,512,099                   & 833,081            & 81,989,133         & 26,683                                  & 936,468,800                            & 91,968,444,615                         & 1,290,680                              \\ \hline
\end{tabular}
}
\end{table}

\noindent\textbf{\textit{Baselines}}: We compare the performance of PBNG against the following baselines:
\begin{itemize}[leftmargin=*]
    \item \texttt{BUP}$\rightarrow$ sequential bottom-up peeling~(alg.\ref{alg:bottomup}) that does not use BE-Index.
    \item \texttt{ParB}$\rightarrow$P\textsc{ar}B\textsc{utterfly} framework\footnote{\url{https://github.com/jeshi96/parbutterfly}} with the best performing BatchS aggregation 
    method~\cite{shiParbutterfly}. It parallelizes each iteration of \texttt{BUP} using a parallel bucketing structure~\cite{julienne}. 
    \item \texttt{BE\_Batch}~(for wing decomposition only)$\rightarrow$ BE-Index assisted peeling with batch processing 
    optimization~\cite{wangBitruss}, and dynamic deletion of bloom-edge links~(sec.\ref{sec:optimizations}).
    \item \texttt{BE\_PC}~(for wing decomposition only)$\rightarrow$ BE-Index assisted progressive compression peeling approach, proposed 
    in~\cite{wangBitruss}.
    It generates candidate subgraphs top-down in the hierarchy to avoid support updates from peeling edges in lower subgraphs~(small 
    $\theta_e$) to edges in higher subgraphs~(high $\theta_e$). Scaling parameter for support threshold of candidate subgraphs is set to 
    $\tau=0.02$, as specified in~\cite{wangBitruss}.
\end{itemize}
Furthermore, to evaluate the effect of optimizations~(sec.\ref{sec:optimizations}), we create two variants of PBNG:
\begin{itemize}[leftmargin=*]
    \item PBNG-\ $\rightarrow$ PBNG without dynamic graph updates~(sec.\ref{sec:dynamic}).
    \item PBNG-{}-\ $\rightarrow$ PBNG without dynamic graph updates~(sec.\ref{sec:dynamic}) and batch processing optimization~(sec.\ref{sec:batch}).
\end{itemize}

\begin{figure}[htbp]
    \centering
\includegraphics[width=\linewidth]{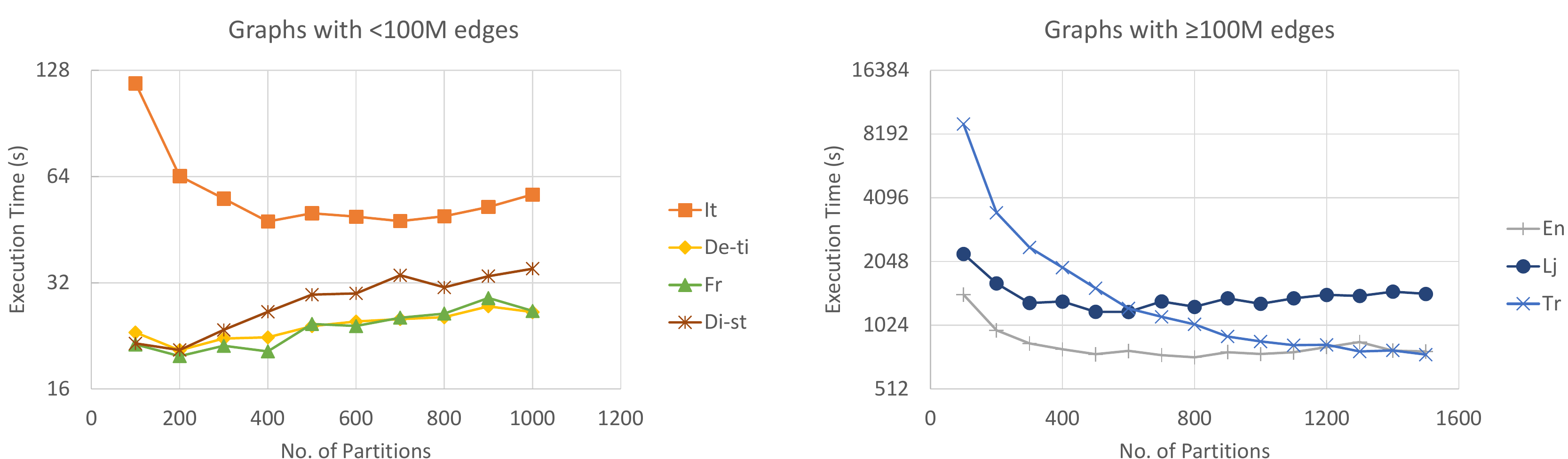}
\caption{Execution time of wing decomposition vs number of partitions $P$ in PBNG.}
    \label{fig:partsWing}
\end{figure}

\noindent\textbf{{Parameter Setting}}: The only user-specified parameter in PBNG is the no. of partitions $P$. For tip decomposition,
we use $P=150$ which was empirically determined in~\cite{lakhotia2020receipt}. For wing decomposition, we measure the runtime
of PBNG as a function of $P$ as shown in fig.\ref{fig:partsWing}. 
Performance of PBNG CD improves with a decrease in $P$ because of reduced peeling iterations
and larger peeling set~(batch size) per iteration. However, 
for PBNG FD, a small value of $P$ reduces parallelism and 
increases the workload. 
Thus, $P$ represents a trade-off between the two phases of PBNG.
Based on our our observations, we set $P=400$  for graphs with $<100$ M edges, and $P=1000$ for graphs with $\geq 100$ M edges. 
We also note that the performance of PBNG is robust~(within $2\times$ of the optimal) in a wide range of $P$ for 
both small and large \looseness=-1datasets.

\subsection{Results: Wing Decomposition}\label{sec:resultsWing}

\subsubsection{Comparison with Baselines}\label{sec:baselineWing}
Table~\ref{table:performanceWing} shows a detailed comparison of PBNG and baseline wing decomposition algorithms. 
To compare the workload of different algorithms, we measure the number of support updates applied in each algorithm~\cite{wangBitruss}.
Note that this may under-represent the workload of \texttt{BUP} and \texttt{ParB}, as they cannot retrieve affected
edges during peeling in constant time. However, it is a useful metric to compare BE-Index based approaches~\cite{wangBitruss}
as support updates represent bulk of the computation performed by during decomposition. 

Amongst the baseline algorithms, \texttt{BE\_PC} demonstrates state-of-the-art execution time and lowest computational 
workload, represented by the number of support updates, due to its top-down subgraph construction approach. 
However, with the two-phased peeling and batch optimizations, support updates in PBNG are at par
or even lower than \texttt{BE\_PC} in some cases. Moreover, most updates in PBNG are applied to a simple array and are relatively
cheaper compared to updates applied to priority queue data structure in all baselines~(including~\texttt{BE\_PC}). 
Furthermore, by utilizing parallel computational resources, PBNG achieves up to $38.5\times$ speedup over \texttt{BE\_PC},
with especially high speedup on large \looseness=-1datasets.

Compared to the parallel framework $\texttt{ParB}$, PBNG is two orders of magnitude or more~(up to $295\times$) faster. 
This is because $\texttt{ParB}$ does not use BE-Index for efficient peeling, does not utilize batch optimizations to reduce computations, and requires large amount of parallel peeling iterations~($\rho$). 
The number of threads synchronizations is directly proportional to $\rho$\footnote{For $\texttt{ParB}$, $\rho$ can be determined by counting peeling iterations in PBNG FD, even if $\texttt{ParB}$ itself is unable
to decompose the graph.}, and PBNG achieves up to $15260\times$ 
reduction in $\rho$ compared to $\texttt{ParB}$. This is primarily 
because PBNG CD peels vertices with a broad range of support in 
every iteration, and PBNG FD does not require a global thread 
synchronization at alll. This drastic reduction in $\rho$ is the 
primary contributor to PBNG's parallel efficiency.

Quite remarkably, PBNG is the \emph{only algorithm to successfully} wing decompose \emph{Gtr} and \emph{De-ut} datasets in few hours,
whereas all of the baselines fail to decompose them in two days\footnote{None of the algorithms could wing decompose \emph{Or} dataset -- \texttt{BE\_Batch}, \texttt{BE\_PC} and PBNG due to large memory requirement of BE-Index, and \texttt{BUP} and \texttt{ParB} due
to infeasible amount of execution time}. Overall, PBNG achieves dramatic reduction in workload, execution
time and synchronization compared to all previously existing algorithms.

\begin{table}[]
\caption{Comparing execution time (\textbf{$\mathbf{t}$}), number of support updates and thread synchronization~(or parallel peeling iterations $\rho$) for PBNG and baseline algorithms. Missing entries denote that execution did not finish in $2$ days. \texttt{ParB} will generate same number of support updates as \texttt{BUP}, and parallel variants of all baselines will have same amount of synchronization~($\rho$)
as \texttt{ParB}.}
\label{table:performanceWing}
\resizebox{\linewidth}{!}{%
\begin{tabular}{cc|c|c|c|c|c|c|c|c|c|c|c|}
\cline{3-13}
\textbf{}                                                                                                     & \textbf{}                           & \textbf{Di-af}                        & \textbf{De-ti}                        & \textbf{Fr}                           & \textbf{Di-st}                        & \textbf{It}                           & \textbf{Digg}                          & \textbf{En}                            & \textbf{Lj}                            & \textbf{Gtr}                           & \textbf{Tr}                           & \textbf{De-ut}                         \\ \hline\hline
\multicolumn{1}{|c|}{}                                                                                        & \texttt{BUP}       & 911                                   & 6,622                                  & 2,565                                  & 9,972                                  & 38,579                                 & -                                      & -                                      & -                                      & -                                      & -                                     & -                                      \\ \cline{2-13} 
\multicolumn{1}{|c|}{}                                                                                        & \texttt{ParB}      & 324                                   & 3,105                                  & 1,434                                  & 2,976                                  & 14,087                                 & -                                      & -                                      & -                                      & -                                      & -                                     & -                                      \\ \cline{2-13} 
\multicolumn{1}{|c|}{}                                                                                        & \texttt{BE\_Batch} & 87                                    & 568                                   & 678                                   & 961                                   & 9,940                                  & -                                      & 78,847                                  & -                                      & -                                      & 54,865                                 & -                                      \\ \cline{2-13} 
\multicolumn{1}{|c|}{}                                                                                        & \texttt{BE\_PC}    & 56                                  & 312                                   & 237                                   & 314                                   & 1756                                  & 37,006                                  & 25,905                                  & 50,901                                  & -                                      & 28,418                                 & -                                      \\ \cline{2-13} 
\multicolumn{1}{|c|}{\multirow{-5}{*}{\textbf{$\mathbf{t}$(s)}}}                                              & PBNG      & {\color[HTML]{009901} \textbf{7.1}}   & {\color[HTML]{009901} \textbf{22.4}}  & {\color[HTML]{009901} \textbf{20.4}}  & {\color[HTML]{009901} \textbf{26.5}}  & {\color[HTML]{009901} \textbf{47.7}}    & {\color[HTML]{009901} \textbf{960}}   & {\color[HTML]{009901} \textbf{748}}    & {\color[HTML]{009901} \textbf{1,293}}   & {\color[HTML]{009901} \textbf{13,253}}  & {\color[HTML]{009901} \textbf{858}}  & {\color[HTML]{009901} \textbf{6,661}}   \\ \hline\hline
\multicolumn{1}{|c|}{}                                                                                        & \texttt{BUP}       & 4.7                                   & 37.3                                  & 39.4                                  & 122                                   & 424                                   & -                                      & -                                      & -                                      & -                                      & -                                     & -                                      \\ \cline{2-13} 
\multicolumn{1}{|c|}{}                                                                                        & \texttt{BE\_Batch} & 2.9                                   & 16.6                                  & 24.8                                  & 33.9                                  & 119.0                                 & 1,413                                   & 1,116                                   & -                                      & -                                      & 679                                   & -                                      \\ \cline{2-13} 
\multicolumn{1}{|c|}{}                                                                                        & \texttt{BE\_PC}    & {\color[HTML]{009901} \textbf{1.1}}   & {\color[HTML]{009901} \textbf{5.7}}   & {\color[HTML]{009901} \textbf{7.6}}   & {\color[HTML]{009901} \textbf{9.9}}   & 33.6                                  & {\color[HTML]{009901} \textbf{592}}    & {\color[HTML]{009901} \textbf{391}}    & 785                                    & -                                      & 390                                   & -                                      \\ \cline{2-13} 
\multicolumn{1}{|c|}{\multirow{-4}{*}{\textbf{\begin{tabular}[c]{@{}c@{}}Updates\\ (billions)\end{tabular}}}} & PBNG      & 1.9                                   & 9                                   & 11                                  & 11.6                                  & {\color[HTML]{009901} \textbf{26.9}}  & 794                                    & 402                                    & {\color[HTML]{009901} \textbf{678}}    & {\color[HTML]{009901} \textbf{5,765}}   & {\color[HTML]{009901} \textbf{164}}   & {\color[HTML]{009901} \textbf{5,530}}   \\ \hline\hline
\multicolumn{1}{|c|}{}                                                                                        & \texttt{ParB}       & 179,177                               & 868,527                               & 181,114                               & 1,118,178                             & 781,955                               & 1,077,389                              & 8,456,797                              & 4,338,205                              & 3,655,765                              & 31,043,711                            & 4,844,812                              \\ \cline{2-13} 
\multicolumn{1}{|c|}{\multirow{-2}{*}{$\mathbf{\rho}$}}                                                       & PBNG      & {\color[HTML]{009901} \textbf{3,666}} & {\color[HTML]{009901} \textbf{5,902}} & {\color[HTML]{009901} \textbf{4,062}} & {\color[HTML]{009901} \textbf{4,442}} & {\color[HTML]{009901} \textbf{4,037}} & {\color[HTML]{009901} \textbf{14,111}} & {\color[HTML]{009901} \textbf{16,324}} & {\color[HTML]{009901} \textbf{19,824}} & {\color[HTML]{009901} \textbf{18,371}} & {\color[HTML]{009901} \textbf{2,034}} & {\color[HTML]{009901} \textbf{15,136}} \\ \hline
\end{tabular}
}
\end{table}

\subsubsection{Effect of Optimizations}
Since dynamic BE-Index updates do not affect the updates generated during peeling, PBNG and PBNG- exhibit the same number of support updates.
Hence, to highlight the benefits of BE-Index updates~(sec.\ref{sec:dynamic}), we also measure the number of bloom-edge links traversed in PBNG with and without the optimizations.
Fig.\ref{fig:optWing} shows the effect of optimizations on the performance of PBNG. 

Normalized performance of PBNG-~(fig.\ref{fig:optWing}) shows that deletion of bloom-edge links~(corresponding to peeled edges) from BE-Index 
reduces traversal by an average of $1.4\times$, and execution time by average $1.11\times$. 
However, traversal is relatively inexpensive compared to support updates as the latter involve atomic computations on array elements~(PBNG CD) or on a priority queue~(PBNG FD).
Fig.\ref{fig:optWing} also clearly shows a direct correlation between the execution time of PBNG-{}- and the number of support updates. 
Consequently, the performance is drastically impacted by batch processing, without which large datasets $\emph{Gtr}$, $\emph{Tr}$ and $\emph{De-ut}$ can not
be decomposed in two days by PBNG-{}-. Normalized performance of PBNG-{}- shows 
that both optimizations cumulatively enable an average reduction of $9.1\times$ and $21\times$ in the number of support updates and execution time of wing decomposition,
respectively. 
This shows that the two-phased approach of PBNG is highly suitable for batch optimization as it peels large number of edges per parallel iteration.

\begin{figure}[htbp]
    \centering
\includegraphics[width=\linewidth]{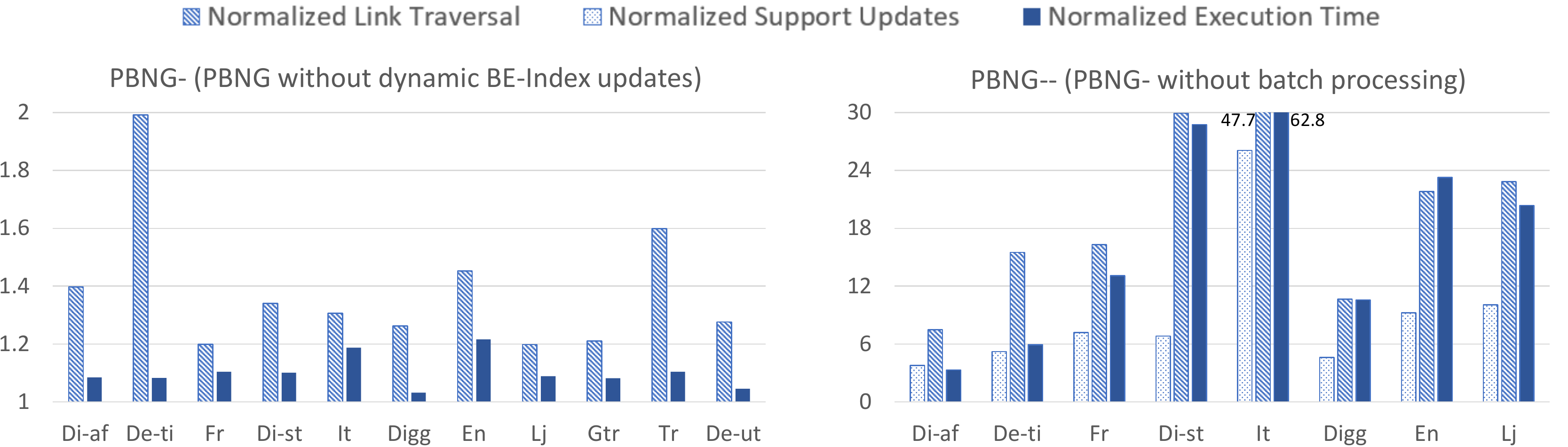}
\caption{Effect of optimizations~(sec.\ref{sec:optimizations}) on wing decomposition in PBNG.
All quantities are normalized with respective measurements for PBNG with all optimizations enabled. With batch processing disabled~(PBNG-{}-), \emph{Gtr, Tr} and \emph{De-ut} did not finish within $2$ days.}
    \label{fig:optWing}
\end{figure}

\subsubsection{Comparison of Different Phases}
\begin{figure}[htbp]
    \centering
\includegraphics[width=\linewidth]{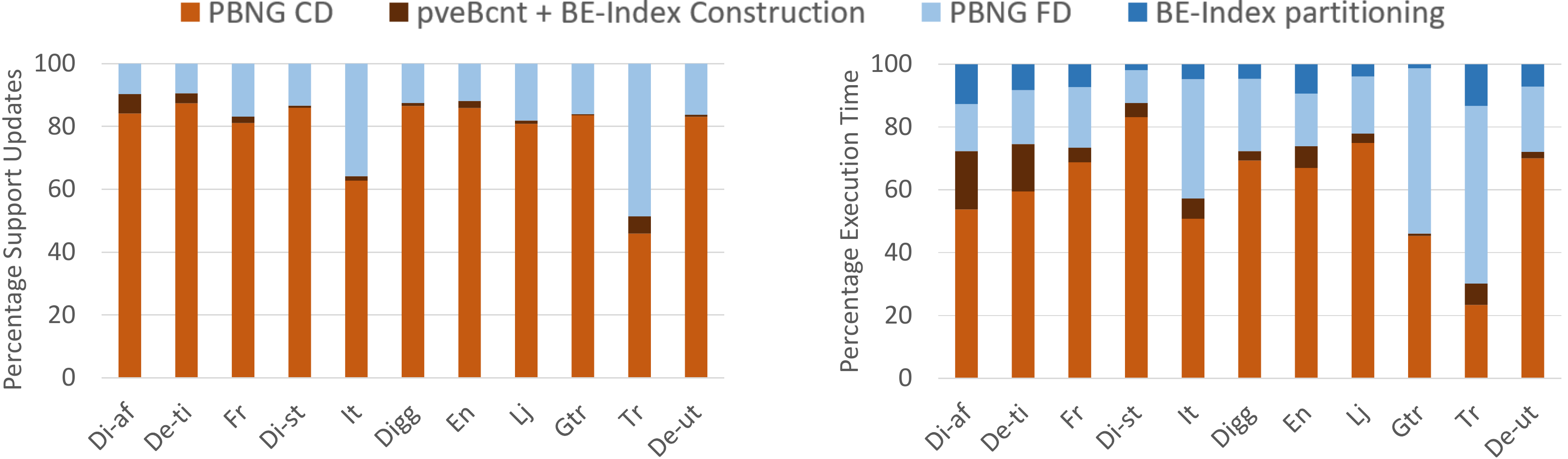}
\caption{Contribution of different steps to the overall support updates and the execution time of wing decomposition in PBNG}
    \label{fig:stepsWing}
\end{figure}

Fig.\ref{fig:stepsWing} shows a breakdown of the support updates and execution time of PBNG across different steps, namely initial butterfly 
counting and BE-Index Construction, peeling in PBNG CD,  BE-Index partitioning\footnote{BE-Index partitioning does not update the support of 
edges.} and peeling in PBNG FD.
For most datasets, PBNG CD dominates the overall workload, 
contributing more than $60\%$ of the support updates for most 
graphs. In some datasets such as \emph{Tr} and \emph{Gtr}, the
batch optimizations drastically reduce the workload of PBNG CD, 
rendering PBNG FD as the dominant phase.
The trends in execution time are largely similar to those of support updates. However, due to
differences in parallel scalability of different steps, contribution of PBNG FD to execution time of several datasets is slightly 
higher than its corresponding contribution to support updates. We also observe that peeling in PBNG CD and PBNG FD is 
much more expensive compared to BE-Index construction and partitioning.

\subsubsection{Scalability}
\begin{figure}[htbp]
    \centering
\includegraphics[width=\linewidth]{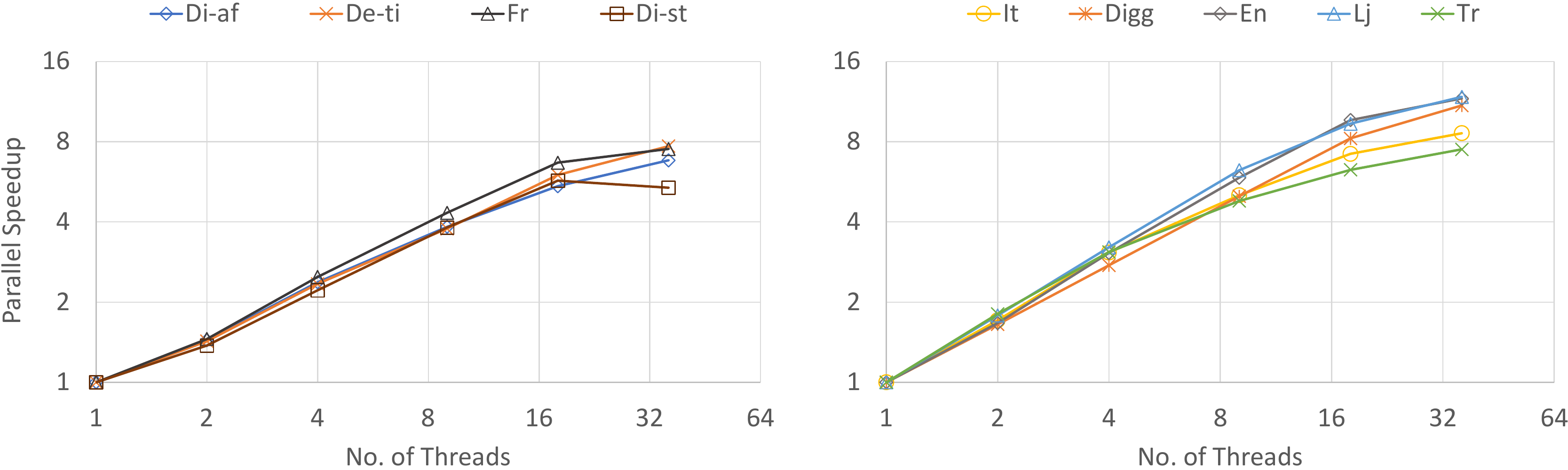}
\caption{Strong scaling of wing decomposition in PBNG. Datasets
shown on the left were decomposed in less than a minute, and on the
right take several minutes.}
    \label{fig:scaleWing}
\end{figure}
Fig.\ref{fig:scaleWing} demonstrates the parallel speedup of PBNG over sequential
execution\footnote{We also compared a serial implementation of PBNG~(for both
wing and tip decomposition) with no synchronization primitives~(atomics) and 
sequential implementations of kernels such as prefix scan. However, the observed 
performance difference between such implementation and single-threaded execution 
of parallel PBNG was negligible.}. Overall, PBNG provides an average $8.7\times$ 
parallel speedup with $36$ threads, which is significantly better than $\texttt{ParB}$. 
Furthermore, the speedup is generally higher
for large datasets~(up to $11.8\times$ for \emph{Lj}), which are
highly time consuming. 
Contrarily, $\texttt{ParB}$ achieves an average $2.6\times$ speedup over sequential $\texttt{BUP}$, due to the large amount of synchronization.

We also observe that PBNG \emph{consistently accelerates} 
decomposition up to $18$ threads~(single socket), providing average $7.2\times$ 
parallel speedup. However, scaling to two sockets~(increasing 
threads from $18$ to $36$) only fetches $1.2\times$ speedup on 
average. This could be due to NUMA effects which increases the
cost of memory accesses and atomics. This can significantly impact 
the performance as PBNG's workload is dominated by traversal of the 
large BE-Index and atomic support updates.
Further, edges contained in a large number of butterflies
may receive numerous support updates, which increases
coherency traffic and reduces scalability. 

\subsection{Results: Tip Decomposition}\label{sec:resultsTip}
To evaluate tip decomposition in PBNG, we select $6$ of the largest 
datasets from table~\ref{table:datasets} and individually decompose both
vertex sets in them. Without loss of generality, we label the vertex set
with higher peeling complexity as $U$ and the other as $V$. 
Corresponding to the set being decomposed, we suffix the dataset name with $U$ or $V$.

\subsubsection{Comparison with Baselines}
\begin{table}[htbp]
\centering
\caption{Comparing execution time ($t$), \# wedges traversed and \# synchronization rounds (or parallel peeling iterations $\rho$) of PBNG and baseline algorithms for tip decomposition. \texttt{ParB} traverses the same \# wedges as \texttt{BUP} and has missing entries due to out-of-memory error. \texttt{BUP} could not decompose \emph{TrU} in $2$ days.}
\label{table:performance}
\resizebox{\linewidth}{!}{%
\begin{tabular}{cc|c|c|c|c|c|c|c|c|c|c|c|l|}
\cline{3-14}
\textbf{}                                                                                                                 & \textbf{}                      & \textbf{EnU}                          & \textbf{EnV}                         & \textbf{LjU}                          & \textbf{LjV}                         & \textbf{GtrU}                    & \textbf{GtrV}                    & \textbf{TrU}                          & \textbf{TrV}                          & \textbf{OrU}                          & \textbf{OrV}                          & \textbf{De-utU}                       & \textbf{De-utV}                      \\ \hline\hline
\multicolumn{1}{|c|}{}                                                                                                    & \texttt{BUP}  & 111,777                               & 281                                  & 67,588                                & 200                                  &  12,036                                &          221                        & -                                     & 5,711                                 & 39,079                                & 2,297                                 & 12,260                                & 428                                  \\ \cline{2-14} 
\multicolumn{1}{|c|}{}                                                                                                    & \texttt{ParB} & -                                     & 198                                  & -                                     & 132.5                                & -                                 &    163.9                              & -                                     & 3,524                                 & -                                     & 1,510                                 & -                                     & 377.7                                \\ \cline{2-14} 
\multicolumn{1}{|c|}{\multirow{-3}{*}{\textbf{$\mathbf{t}$(s)}}}                                                          & PBNG & {\color[HTML]{009901} \textbf{1,383}} & {\color[HTML]{009901} \textbf{31.1}} & {\color[HTML]{009901} \textbf{911.1}} & {\color[HTML]{009901} \textbf{23.7}} & {\color[HTML]{009901} \textbf{163.9}} & {\color[HTML]{009901} \textbf{26.5}} & {\color[HTML]{009901} \textbf{2,784}} & {\color[HTML]{009901} \textbf{530.6}} & {\color[HTML]{009901} \textbf{1,865}} & {\color[HTML]{009901} \textbf{136}}   & {\color[HTML]{009901} \textbf{402.4}} & {\color[HTML]{009901} \textbf{32.4}} \\ \hline\hline
\multicolumn{1}{|c|}{}                                                                                                    & \texttt{BUP}  & 12,583                                & 29.6                                 & 5,403                                 & 14.3                                 &          3,183                        &       26.6                           & 211,156                               & 1,740                                 & 4,975                                 & 231.4                                 & 2,861                                 & 70.1                                 \\ \cline{2-14} 
\multicolumn{1}{|c|}{\multirow{-2}{*}{\textbf{\begin{tabular}[c]{@{}c@{}}Wedges\\ (billions)\end{tabular}}}} & PBNG & {\color[HTML]{009901} \textbf{2,414}} & {\color[HTML]{009901} \textbf{22.2}} & {\color[HTML]{009901} \textbf{1,003}} & {\color[HTML]{009901} \textbf{11.7}} & {\color[HTML]{009901} \textbf{1,526}} &    {\color[HTML]{009901} \textbf{28.2}}             & {\color[HTML]{009901} \textbf{3,298}} & {\color[HTML]{009901} \textbf{658.1}} & {\color[HTML]{009901} \textbf{2,728}} & {\color[HTML]{009901} \textbf{170.4}} & {\color[HTML]{009901} \textbf{1,503}} & 51.3                                 \\ \hline\hline
\multicolumn{1}{|c|}{}                                                                                                    & \texttt{ParB} & 1,512,922                             & 83,800                               & 1,479,495                             & 83,423                               &           491,192                       &             73,323                     & 1,476,015                             & 342,672                               & 1,136,129                             & 334,064                               & 670,189                               & 127,328                              \\ \cline{2-14} 
\multicolumn{1}{|c|}{\multirow{-2}{*}{$\mathbf{\rho}$}}                                                                   & PBNG & {\color[HTML]{009901} \textbf{1,724}} & {\color[HTML]{009901} \textbf{453}}  & {\color[HTML]{009901} \textbf{1,477}} & {\color[HTML]{009901} \textbf{456}}  & {\color[HTML]{009901} \textbf{2,062}} & {\color[HTML]{009901} \textbf{362}} & {\color[HTML]{009901} \textbf{1,335}} & {\color[HTML]{009901} \textbf{1,381}} & {\color[HTML]{009901} \textbf{1,160}} & {\color[HTML]{009901} \textbf{639}}   & {\color[HTML]{009901} \textbf{1,113}} & {\color[HTML]{009901} \textbf{406}}  \\ \hline
\end{tabular}
}
\end{table}

Table~\ref{table:performance} shows a detailed comparison of various tip 
decomposition algorithms. 
To compare the workload of tip decomposition algorithms, we measure the
number of wedges traversed in $G$. Wedge traversal is required to compute 
butterflies between vertex pairs during counting/peeling, and represents 
bulk of the computation performed in tip decomposition .

With up to $80.8\times$ and 
$64.7\times$ speedup over \texttt{BUP} and \texttt{ParB}, 
respectively, PBNG is dramatically faster than the baselines, for \textit{all} 
datasets. Contrarily, \texttt{ParB} achieves a maximum $1.6\times$ speedup 
compared to sequential \texttt{BUP} for \textit{TrV}
dataset. The speedups are typically higher for large datasets that offer large amount of computation to parallelize and benefit more from batch optimization~(sec.\ref{sec:batch}). 

Optimization benefits are also evident in the wedge traversal
of PBNG\footnote{Wedge traversal by \texttt{BUP} can be computed without executing alg.\ref{alg:bottomup}, by simply aggregating the $2$-hop neighborhood size of vertices in $U$ or $V$.}. For 
\textit{all} datasets, PBNG traverses fewer wedges than the  
baselines, achieving up to $64\times$ reduction in 
wedges traversed. Furthermore, PBNG achieves up to $1105\times$
reduction in synchronization~($\rho$) over $\texttt{ParB}$, 
due to its two-phased peeling approach. 
The resulting increase in parallel efficiency and workload 
optimizations enable PBNG to 
decompose large datasets like \textit{EnU} 
in few minutes, unlike baselines that take few days for the same.

Quite remarkably, PBNG is the \emph{only algorithm to successfully} tip 
decompose \emph{TrU} dataset within an hour,
whereas the baselines fail to decompose it in two days. 
We also note PBNG can decompose both vertex sets of $\texttt{Or}$ dataset in approximately half an hour, even though none of the algorithms could
feasibly wing decompose the $\texttt{Or}$ dataset~(sec.\ref{sec:resultsWing}). 
Thus, tip decomposition is advantageous over wing decomposition,
in terms of efficiency and feasibility.

\subsubsection{Optimizations}
\begin{figure}[htbp]
    \centering
\includegraphics[width=\linewidth]{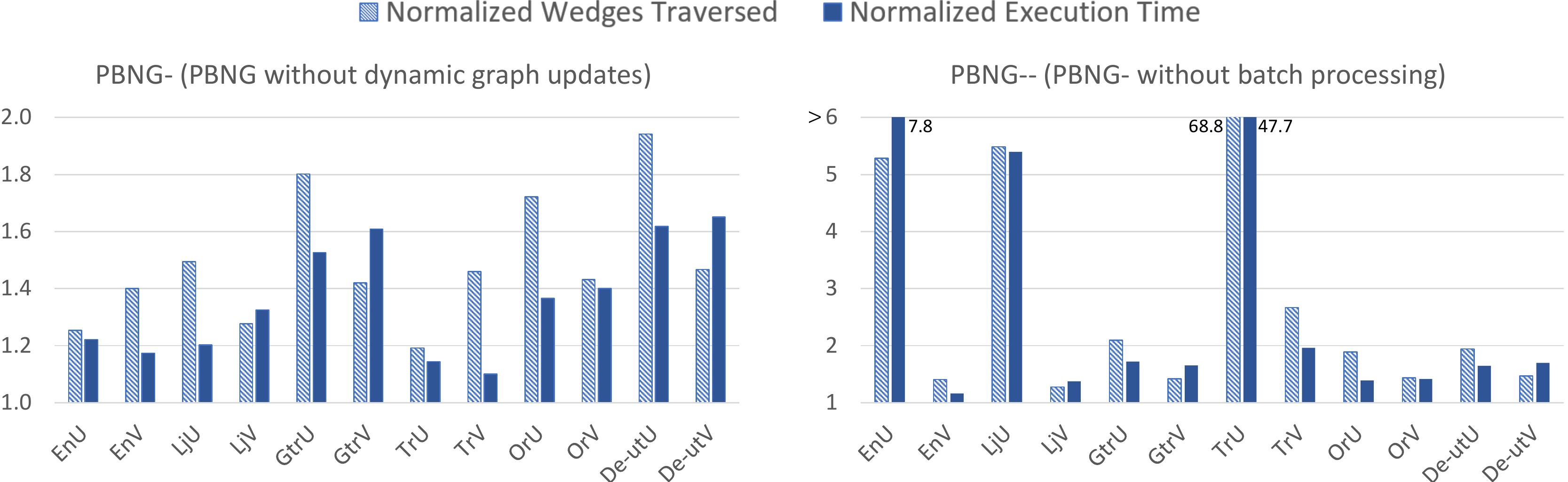}
\caption{Effect of optimizations~(sec.\ref{sec:optimizations}) on tip
decomposition in PBNG.
All quantities are normalized with the respective measurements for PBNG with all optimizations enabled.}
    \label{fig:optTip}
\end{figure}
Fig.\ref{fig:optTip} shows the effect of workload optimizations on tip 
decomposition in PBNG. Clearly, the execution time closely follows the 
variations in number of wedges traversed.

Dynamic deletion of edges~(corresponding to peeled vertices) from adjacency lists
can potentially half the wedge workload since each wedge has two endpoints in 
peeling set. Normalized performance of PBNG-~(fig.\ref{fig:optWing}) shows that 
it achieves $1.41\times$ and $1.29\times$ average reduction in wedges and execution time, \looseness=-1respectively.

Similar to wing decomposition, the batch optimization provides dramatic
improvement in workload and execution time. This is especially true for 
datasets with a large ratio of total wedges with endpoints in peeling set to 
the wedges traversed during counting~(for example, for \emph{LjU, EnU} and \emph{TrU}, this ratio is $>1000$). 
For instance, in \textit{TrU}, both optimizations cumulatively enable 
\textbf{$68.8\times$} and \textbf{$47.7\times$} reduction in wedge traversal and 
execution time, respectively. 
Contrarily, in datasets with small value of this ratio such as \textit{DeV, OrV, LjV} and \emph{EnV}, none of the peeling iterations in PBNG CD utilize re-counting. Consequently, performance of PBNG- and PBNG-{}- is similar for 
these datasets.

\subsubsection{Comparison of phases}
\begin{figure}[htbp]
    \centering
\includegraphics[width=\linewidth]{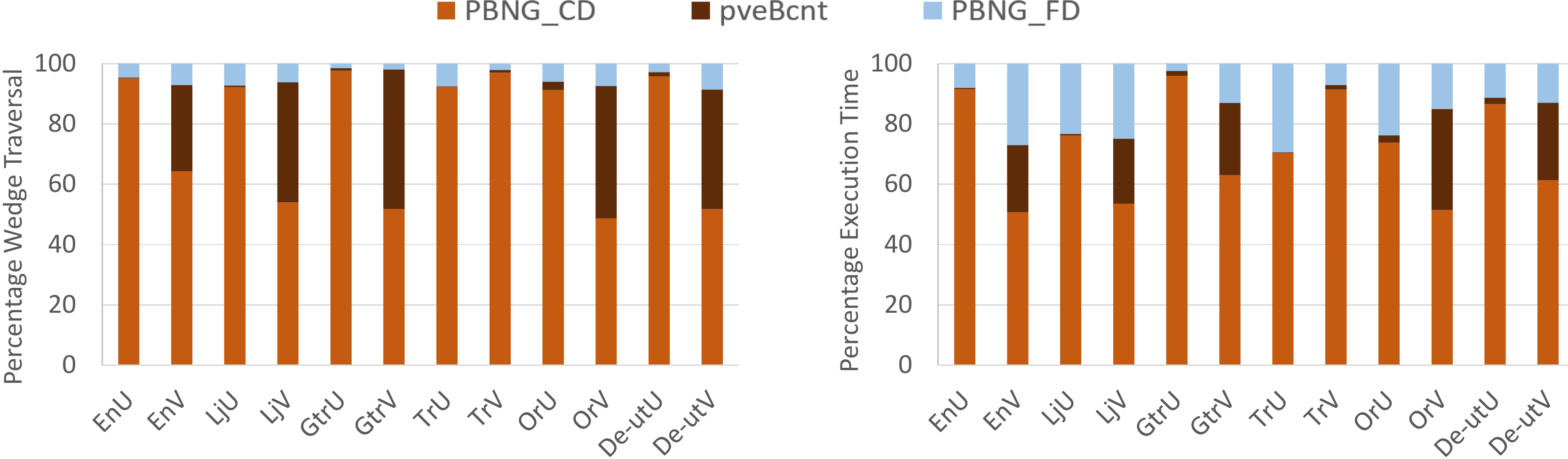}
\caption{Contribution of different steps to the overall wedge traversal and the execution time of tip decomposition in PBNG}
    \label{fig:stepsTip}
\end{figure}

Fig.\ref{fig:stepsWing} shows a breakdown of the wedge traversal and execution 
time of PBNG across different steps, namely initial butterfly 
counting, peeling in PBNG CD and PBNG FD.
As expected, PBNG FD only contributes less than $15\%$ of the total wedge
traversal in tip decomposition. This is because it operates on small subgraphs
that preserve very few wedges of $G$.
When peeling the large workload $U$ vertex set, more than
$80\%$ of the wedge traversal and $70\%$ of the execution time is spent in PBNG 
CD.

\subsubsection{Scalability}
\begin{figure}[htbp]
    \centering
\includegraphics[width=\linewidth]{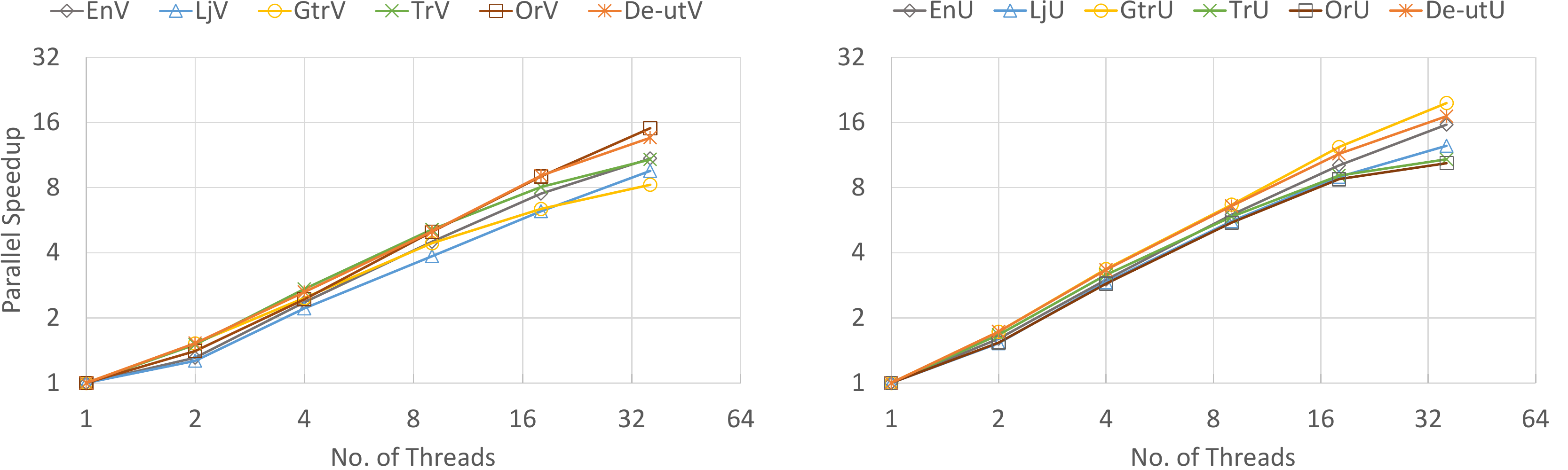}
\caption{Strong scaling of tip decomposition in PBNG}
    \label{fig:scaleTip}
\end{figure}

Fig.\ref{fig:scaleTip} demonstrates the parallel speedup of PBNG over sequential
execution. When peeling the large workload vertex set $U$, PBNG achieves almost 
linear scalability with $14.4\times$ average parallel speedup on $36$ threads, 
and up to $19.7\times$ speedup for $GtrU$ dataset.
Contrarily, $\texttt{ParB}$ achieves an average $1.54\times$ speedup over 
sequential $\texttt{BUP}$, and up to $2.3\times$ speedup for $TrV$ dataset.

Typically, datasets with small amount of wedges (\textit{LjV, EnV}) exhibit 
lower speedup,  because they provide lower workload per synchronization round on 
average. For example, \textit{LjV} traverses $86\times$ fewer wedges than 
\textit{LjU} but incurs only $3.2\times$ fewer synchronizations. This increases 
the relative overheads of parallelization and restricts the parallel scalability 
of PBNG CD, which is the highest workload step in PBNG~(fig.\ref{fig:stepsTip}). 
Similar to wing decomposition, NUMA effects on scalability to multiple sockets~($36$ threads) can be seen in tip decomposition of some datasets such as 
$\emph{OrU}$ and $\emph{TrU}$. However, we still observe $1.4\times$ average
speedup on large datasets when threads are increased from $18$ to $36$ threads.
This is possibly because most of the workload in tip decomposition is comprised 
of wedge traversal which only reads the graph data structure. It incurs much 
fewer support updates, and in turn atomic writes, compared to wing decomposition.

\section{Related Work}
Discovering dense subgraphs and communities in networks is a key operation in several applications~\cite{anomalyDet, spamDet, communityDet, fang2020effective, otherapp1, otherapp2, nathan2017local, staudt2015engineering, riedy2011parallel}.
Motif-based techniques are widely used to reveal dense regions in  graphs~\cite{sariyuce2016fast, fang2019efficient, gibson2005discovering, sariyuce2018local, angel2014dense, trussVLDB, lee2010survey, coreVLDB, coreVLDBJ, wang2018efficient, PMID:16873465, trussVLDB, tsourakakis2017scalable,tsourakakis2014novel, sariyucePeeling, aksoy2017measuring, wang2010triangulation}. 
Motifs like triangles represent a quantum of cohesion in graphs and
the number of motifs containing an entity~(vertex or an edge) acts as 
an indicator of its local density. Consequently, several recent works have 
focused on efficiently finding such motifs in the graphs~\cite{green2018logarithmic, shun2015multicore,ahmed2015efficient, shiParbutterfly, wangButterfly, hu2018tricore, fox2018fast, ma2019linc}.

Nucleus decomposition is a clique based technique for discovering hierarchical 
dense regions in unipartite graphs. 
Instead of per-entity clique count in the entire graph, 
it considers the minimum clique count of a subgraph's entities 
as an indicator of that subgraph's density~\cite{10.1007/s10115-016-0965-5}.
This allows mining denser subgraphs compared to  
counting alone~\cite{10.1007/s10115-016-0965-5, sariyuce2015finding}. 
Truss decomposition is a special and one of the most popular cases of 
nucleus decomposition that uses triangle clique.
It is a part of the GraphChallenge~\cite{samsi2017static} initiative, that has resulted in highly scalable parallel decomposition algorithms~\cite{date2017collaborative,voegele2017parallel,smith2017truss,green2017quickly}. However, nucleus decomposition is not applicable on
bipartite graphs as they do not have \looseness=-1cliques.

The simplest non-trivial motif in a bipartite graph is a Butterfly~(2,2-biclique). Several algorithms for butterfly counting have been developed: in-memory or external memory~\cite{wangButterfly,wangRectangle}, exact or approximate counting~\cite{sanei2018butterfly,sanei2019fleet} and parallel counting on various platforms~\cite{shiParbutterfly, wangButterfly, wangRectangle}. 
The most efficient approaches are based on Chiba and Nishizeki's~\cite{chibaArboricity} vertex-priority butterfly counting algorithm. Wang et al.\cite{wangButterfly} propose a cache optimized variant 
of this algorithm and use shared-memory parallelism for acceleration. Independently, Shi et al.\cite{shiParbutterfly} develop provably efficient shared-memory parallel implementations of this algorithm. 
Notably, their algorithms are able to extract parallelism at the 
granularity of wedges explored. Such fine-grained parallelism can
also be explored for improving parallel scalability of PBNG.

Inspired by $k$-truss, Sariyuce et al.\cite{sariyucePeeling} defined $k$-tips 
and $k$-wings as subgraphs with minimum $k$ butterflies incident on every 
vertex and edge, respectively. Similar to nucleus decomposition
algorithms, they designed \emph{bottom-up peeling} algorithms to find hierarchies 
of $k$-tips and $k$-wings. Independently, Zou~\cite{zouBitruss} defined the
notion of bitruss similar to $k$-wing. 
Shi et al.\cite{shiParbutterfly} propose the $P\textsc{ar}B\textsc{utterfly}$ 
framework that parallelizes individual peeling iterations.

Chiba and Nishizeki~\cite{chibaArboricity} proposed that
wedges traversed in butterfly counting algorithm
can act as a space-efficient representation of all
butterflies in the graph. Wang et al.\cite{wangBitruss} propose a
butterfly representation called BE-Index~(also derived from the counting algorithm) for efficient support updates during edge peeling. Based on BE-Index, they develop several peeling algorithms for wing decomposition that
achieve state-of-the-art computational efficiency and are used
as baselines in this paper. 

Very recently, Wang et al.\cite{wang2021towards} also proposed parallel 
versions of BE-Index based wing decomposition algorithms. 
Although the code is not publicly available, authors report the performance on few graphs. Based on the reported results, PBNG
significantly outperforms their parallel algorithms as well.
Secondly, these algorithms parallelize individual peeling iterations
and will incur heavy synchronization similar to $P\textsc{ar}B\textsc{utterfly}$~(table~\ref{table:performanceWing}). 
Lastly, they are only designed for wing decomposition, whereas PBNG 
comprehensively targets both wing and tip decomposition. 
This is important because tip decomposition in PBNG  
(a)~is typically faster than wing decomposition, and (b)~can 
process large datasets like \emph{Orkut} in few minutes, that none of 
the existing tip or wing decomposition algorithms can process in 
several days.


\section{Conclusion and Future Work}
In this paper, we studied the problem of 
bipartite graph decomposition which is a computationally demanding 
analytic for which the existing
algorithms were not amenable to efficient parallelization.
We proposed a novel parallelism friendly 
two-phased peeling framework called PBNG, 
that is the first to exploit parallelism
across the levels of decomposition hierarchy. The proposed
approach further enabled novel optimizations that drastically reduce
computational workload, allowing bipartite decomposition
to scale beyond the limits of current practice.

We presented a comprehensive empirical evaluation of PBNG on a 
shared-memory multicore server and showed that it can process 
some of the largest publicly available bipartite datasets two orders of 
magnitude faster than the state-of-the-art. 
PBNG also achieved a dramatic reduction in thread synchronization, allowing
up to $19.7\times$ self-relative parallel speedup on $36$ threads.

There are several directions for future research in the context of this work.
The proposed two-phased peeling can open up avenues for distributed-memory
parallel bipartite decomposition on HPC clusters. Clusters allow 
scaling of both computational and memory resources, which is crucial for 
decomposing large graphs. To further improve the performance of PBNG,
we will also explore fine-grained parallelism, and techniques to avoid atomics and enhance memory access locality. Such 
system optimizations have been shown to be highly effective
in conventional graph processing frameworks~\cite{shun2013ligra,lakhotia2020gpop,zhang2018graphit}. 
We believe that two-phased peeling could also benefit 
nucleus decomposition in unipartite graphs, which is an interesting 
direction to explore.

\begin{small}
{\setlength{\parindent}{0cm}
\paragraph*{\textbf{Acknowledgement}}
This material is based on work supported by the Defense Advanced Research Projects Agency (DARPA) under Contract Number FA8750-17-C-0086, and National Science Foundation (NSF) under Grant Numbers CNS-2009057 and OAC-1911229. Any opinions, findings and conclusions or recommendations	expressed in this material are those of the authors and do not necessarily reflect the views of DARPA or NSF. The U.S. Government is authorized to reproduce and distribute reprints for Government purposes notwithstanding any copyright notation here on.
}
\end{small}

\bibliographystyle{acm}
\bibliography{bibliography} 

\end{document}